\documentclass[3p,preprint]{elsarticle}

\journal{Theoretical Computer Science}

\usepackage{amssymb}

\usepackage[figuresright]{rotating}

\newtheorem{theorem}{Theorem}
\newtheorem{lemma}[theorem]{Lemma}

\newtheorem{corollary}[theorem]{Corollary}

\newtheorem{example}{Example}

\newproof{proof}{Proof}

\usepackage[utf8]{inputenc}
\usepackage[algosection]{algorithm2e}
\usepackage{endnotes}
\usepackage{listings}
\usepackage{mathtools}
\usepackage{mathrsfs}
\usepackage{longtable}
\usepackage{enumitem}
\usepackage{etoolbox}
\usepackage{float}
\usepackage{color}
\usepackage{hyperref}

\usepackage{complexity}

\usepackage{url}
\usepackage{complexity}
\usepackage{verbatim}
\newclass{\DCM}{DCM}
\newclass{\RCM}{RCM}
\newclass{\eDCM}{eDCM}
\newclass{\eNPDA}{eNPDA}
\newclass{\DPDA}{DPDA}
\newclass{\RDPDA}{RDPDA}
\newclass{\PDA}{NPDA}
\newclass{\MAJ}{MAJ}
\newclass{\DCMNE}{DCM_{NE}}
\newclass{\TwoDCM}{2DCM}
\newclass{\NCM}{NCM}
\newclass{\NCPM}{NCPM}
\newclass{\eNCM}{eNCM}
\newclass{\eNQA}{eNQA}
\newclass{\eNSA}{eNSA}
\newclass{\eNPCM}{eNPCM}
\newclass{\eNQCM}{eNQCM}
\newclass{\eNSCM}{eNSCM}
\newclass{\DPCM}{DPCM}
\newclass{\NPCM}{NPCM}
\newclass{\NQCM}{NQCM}
\newclass{\NSCM}{NSCM}
\newclass{\NPDA}{NPDA}
\newclass{\TRE}{TRE}
\newclass{\NFA}{NFA}
\newclass{\DFA}{DFA}
\newclass{\NCA}{NCA}
\newclass{\DCA}{DCA}
\newclass{\DTM}{DTM}
\newclass{\NTM}{NTM}
\newclass{\DLOG}{DLOG}
\newclass{\CFG}{CFG}
\newclass{\ETOL}{ET0L}
\newclass{\EDTOL}{EDT0L}
\newclass{\CFP}{CFP}
\newclass{\ORDER}{O}
\newclass{\MATRIX}{M}
\newclass{\BD}{BD}
\newclass{\LB}{LB}
\newclass{\ALL}{ALL}
\newclass{\decLBD}{decLBD}
\newclass{\StLB}{StLB}
\newclass{\SBD}{SBD}
\newclass{\TCA}{TCA}
\newclass{\RNCSA}{RNCSA}
\newclass{\RDCSA}{RDCSA}
\newclass{\ct}{COUNT}

\newclass{\DCSA}{DCSA}
\newclass{\NCSA}{NCSA}
\newclass{\DCSACM}{DCSACM}
\newclass{\NCSACM}{NCSACM}
\newclass{\NTMCM}{NTMCM}
\newclass{\code}{code}

\newclass{\UFIN}{\LL(IND_{UFIN})}
\newclass{\UFINONE}{\LL(IND_{{UFIN}_1})}
\newclass{\FIN}{\LL(IND_{FIN})}
\newclass{\ILIN}{\LL(IND_{LIN})}
\newclass{\ETOLfin}{\LL(ET0L_{FIN})}

\newsavebox{\spacebox}
\begin{lrbox}{\spacebox}
\verb*! !
\end{lrbox}
\newcommand{\blank}{\usebox{\spacebox}}%

\newcommand{\LL}{{\cal L}}

\DeclareMathOperator{\alp}{alph}

\begin{document}

\begin{frontmatter}

\title{On counting functions and slenderness of languages\tnoteref{t1}\tnoteref{t2}}

\tnotetext[t1]{A portion of this paper (in preliminary form) where the main result
(Theorem \ref{main2})  was only shown for unambiguous reversal-bounded
pushdown automata, and with most proofs missing has appeared in the Proceedings of DLT 2018.}

\tnotetext[t2]{\textcopyright 2018. This manuscript version is made available under the CC-BY-NC-ND 4.0 license \url{http://creativecommons.org/licenses/by-nc-nd/4.0/}}

\author[label1]{Oscar H. Ibarra\fnref{fn1}}
\address[label1]{Department of Computer Science\\ University of California, Santa Barbara, CA 93106, USA}
\ead[label1]{ibarra@cs.ucsb.edu}
\fntext[fn1]{Supported, in part, by
NSF Grant CCF-1117708 (Oscar H. Ibarra).}

\author[label2]{Ian McQuillan\fnref{fn2}}
\address[label2]{Department of Computer Science, University of Saskatchewan\\
Saskatoon, SK S7N 5A9, Canada}
\ead[label2]{mcquillan@cs.usask.ca}
\fntext[fn2]{Supported, in part, by Natural Sciences and Engineering Research Council of Canada Grant 2016-06172 (Ian McQuillan).}

\author[label3]{Bala Ravikumar}
\address[label3]{Department of Computer Science\\
	Sonoma State University, Rohnert Park, CA 94928 USA}
	\ead[label3]{ravikuma@sonoma.edu}

\begin{abstract} 
We study counting-regular languages --- these are languages $L$ for which there is a regular language $L'$ such that the number of strings of length $n$ in $L$ and $L'$ are the same for all $n$.
We show that the languages accepted by unambiguous nondeterministic Turing machines with a one-way read-only input tape and a reversal-bounded worktape are counting-regular. Many one-way acceptors are a special case of this model, such as reversal-bounded deterministic pushdown automata, reversal-bounded deterministic queue automata, and many others, and therefore all languages accepted by these models are counting-regular. 
This result is the best possible in the sense that the claim does not hold for either $2$-ambiguous $\PDA$'s, unambiguous $\PDA$'s with no reversal-bound, and other models. 

We also study closure properties of counting-regular languages, and we study decidability problems in regards to counting-regularity. For example, it is shown that the counting-regularity of even some restricted subclasses of $\PDA$'s is undecidable. Lastly, $k$-slender languages --- where there are at most $k$ words of any length ---
are also studied. Amongst other results, it is shown that it is decidable whether a language in any semilinear full trio is $k$-slender.
\end{abstract}


\end{frontmatter}

\begin{abstract} 

\end{abstract}

\section {Introduction}

This work is concerned with the counting function $f_L(n)$ of a language $L$, equal to the number of strings of length $n$ in $L$.
We say that a language $L$ is {\it counting-regular} if there is a regular language $L'$ such that $f_L(n)$ = $f_{L'}(n)$.  If we can effectively find a deterministic finite automaton ($\DFA$) for $L'$ from a representation of $L$ using standard techniques, then we can efficiently compute $f_L(n)$ even if membership in $L$ can be difficult to answer. We say that $L$ is  {\it strongly counting-regular} if for any regular language $L_1$, $L \cap L_1$ is counting-regular. Being strongly counting-regular will allow us to calculate not only the number of strings of length $n$ in the language, but also the number of strings of length $n$ in $L$ that satisfy additional conditions, such as strings in which some symbol occurs an odd number of times, or strings that avoid some substring etc.
A language family $\LL$ is said to be (strongly) counting-regular if all $L\in \LL$ are (strongly) counting regular.


As a simple motivating example, consider the language $L$ = $\{ x \# y\ |\ x, y \in \{ a, b\}^*, |x|_a = |y|_b \}$. The number of strings of length $n$ in $L$ is $2^{n-1}$ (for $n \geq 1$) which follows after a little bit of work from the series sum $\sum_{i=0}^{n-1} \sum_{k=0}^{i} {i \choose k} {n-i-1 \choose k}$ for $f_L(n)$. (This result can also be established by exhibiting a bijective mapping between the set of strings of length $n$ in  $L$ and the set of binary strings of length $n-1$.)
Therefore this language is a counting-regular language while a simpler looking, and closely related language $L_{eq}$ = $\{ x\ |\ x \in \{ a, b\}^*, |x_a| = |x_b|\}$ is not counting-regular. Is there an `explanation' for the disparate behaviors of these two languages? 
In \cite{R}, it was determined that all languages accepted by one-way deterministic counter machines where the counter is reversal-bounded, are counting-regular. Indeed, $L$ above is in this family. Of
interest is the problem of determining other language families that are (strongly) counting-regular.
It can be observed that there are arbitrarily complex languages (e.g., languages that are not recursively enumerable) that are counting-regular. Therefore, to get some insight into counting-regularity, we should restrict the class of languages, such as to the class of context-free languages $\CFL$'s and some subclasses of $\CFL$'s.


Slender languages (those that are $k$-slender, for some $k$, i.e.\ $f_L(n) \leq k$ for all $n$) were studied in \cite{Salomaa,Ilie} in order to better understand the structure of slender context-free languages ($\CFL$'s) as well as to present decision algorithms for problems such as equivalence and containment when $\CFL$'s are restricted to slender languages. Similar results were shown for the more general class of languages generated by matrix grammars \cite{matrix}. It is natural to determine larger families of slender languages for which equivalence, containment etc. are decidable. 

We briefly state some of the motivations for the study presented in this paper. The counting functions of regular languages are well understood \cite{Berstel2} and there are very efficient algorithms to compute $f_L( n)$ when $L$ is regular. But for other classes, such algorithms are not known
although there are a few exceptions such as linearly constrained languages \cite{Massazza}. Hence, as outlined above, a strongly counting-regular language admits an efficient algorithm for counting the number of strings of length $n$ in it. 
In some applications, we are not necessarily interested in computing exactly the number of strings of length $n$ for a specific $n$, but in asymptotic growth rate of $f_L(n)$ as $n \rightarrow \infty$. In the works \cite{Eisman1,Cordy,Eisman2,R} etc., the issue of how well a non-regular language $L$ can be approximated by a $\DFA$ $M$ asymptotically was studied. The asymptotic approximation bound is the density of the language $L(M) \oplus L$, and so if $L$ is strongly counting-regular, we can compute the approximation bound. Another area of application is the static analysis of programs that involves computing the execution information rate \cite{Cui}. The information rate can be measured by modeling a control-flow of a program using a $\DFA$ or a $\DFA$ augmented with counters. Since the information rate of a regular language is efficiently computable, it is also efficiently computable for languages that are (effectively) counting-regular.

A main result shown here is that the languages
accepted by unambiguous nondeterministic Turing machines ($\NTM$s) with a one-way read-only input and a reversal-bounded worktape are strongly counting-regular. 
Many different families of languages accepted by one-way unambiguous and nondeterministic (and
one-way deterministic) machine models are contained in this family, and are therefore counting-regular.

Counting-regularity is a special case of the notion of a language being commutatively equivalent to a regular language, which requires not only a length preserving bijection to the regular language, but a bijection that preserves the Parikh map. It has been shown that every bounded semilinear language is commutatively equivalent to a regular language (proof split across three papers \cite{FlavioBoundedSemilinearpaper1,FlavioBoundedSemilinearpaper2,FlavioBoundedSemilinear}). Here, we
provide an alternate significantly shorter proof that they are all counting-regular. This is then used to show
that in every semilinear trio (a family closed under $\epsilon$-free homomorphism, inverse homomorphism, intersection with regular languages, and every language has the same Parikh map as some regular language), all bounded languages are counting-regular.
More generally, the counting functions of the bounded languages in every semilinear trio are exactly the same, no matter the family. There are many families for which this holds in the literature, such as Turing machines with a one-way read-only input tape plus a finite-crossing worktape \cite{Harju}, and one-way machines with $k$ pushdown stacks where the machines can write to any of the stacks but only pop from the first non-empty pushdown \cite{multipushdown}, uncontrolled finite-index indexed grammars \cite{LATA2017}, amongst many others discussed in \cite{CIAA2016,Harju}. 

The rest of the paper is organized as follows: in Section \ref{notation}, basic definitions, terminology and notation are presented. In Section \ref{sec:countingregular}, the languages accepted by the Turing machine model above are shown to be strongly counting-regular. We also present some specific language families that are not all counting-regular. 
Next, in Section \ref{bounded}, we consider bounded languages and show that all bounded languages in any semilinear
trio are counting-regular. 
In Section \ref{sec:closure}, closure properties of the class of counting-regular languages (and context-free counting-regular languages) are studied. Then, in Section \ref{sec:decidability}, we present some undecidability results regarding counting-regular languages. The main result is that it is undecidable, given a machine of any of the following types:
(i) $2$-ambiguous $\PDA$'s which make only one
reversal on the stack, (ii) nondeterministic
one-counter machines that make only one
reversal on the counter, and  (iii) $2$-ambiguous
nondeterministic one counter machines,
whether the language accepted is counting-regular.
Then, in Sections \ref{sec:slender1} and \ref{sec:slender2}, decidability properties are discussed for slender languages. It is shown that for any semilinear full trio (where certain closure properties are constructive) $\LL$, it is decidable whether a language $L \in \LL$ is $k$-slender, and containment and equality of two $k$-slender languages in $\LL$ are decidable. It is also shown that in every such family, every $k$-slender language in $\LL$ is counting-regular. 

\section{Basic Terminology and Notation}
\label{notation}

We assume an introductory knowledge in formal language and automata theory \cite{Hopcroft}, including
the definitions of finite automata, context-free languages, Turing machines etc. Next, some notations are given.
An {\em alphabet} is a finite set of symbols. A {\em word} $w$ over $\Sigma$ is any finite sequence of
symbols from $\Sigma$. Given an alphabet $\Sigma$, then $\Sigma^*$ is the
set of all words over $\Sigma$, including the empty word $\epsilon$, and $\Sigma^+$ is the set
of non-empty words over $\Sigma$. A {\em language} $L$
is any subset of $\Sigma^*$. The {\em complement} of a language $L \subseteq \Sigma^*$ with respect to $\Sigma$ is
$\overline{L} = \Sigma^* - L$. Given a word $w \in \Sigma^*$, $w^R$ is the reverse of $w$,
 $w[i]$ is the $i$'th character of $w$, and $|w|$ is the length of $w$.
Given an alphabet $\Sigma = \{a_1, \ldots, a_m\}$ and $a \in \Sigma$, $|w|_a$ is the number of $a$'s in $w$. The {\em Parikh map} of $w$ is
$\psi(w) = ( |w|_{a_1}, \ldots, |w|_{a_m})$, which is extended to the
Parikh map of a language $L$, $\psi(L) = \{\psi(w) \mid w \in L\}$.
Also, $\alp(w) = \{ a \in \Sigma \mid |w|_a>0\}$. 
Given languages $L_1,L_2$, the {\em left quotient} of $L_2$ by $L_1$, $L_1^{-1}L_2 = \{ y \mid xy \in L_2, x \in L_1\}$, and the {\em right quotient} of $L_1$ by $L_2$ is $L_1 L_2^{-1} = \{x \mid xy \in L_1, y \in L_2\}$.

For a language $L$, let $f_L(n)$ be the number of strings of length $n$ in $L$. 
A language $L$ is called {\it counting-regular} if there exists a regular language $L'$ such that for all integers $n \geq 0$, $f_L(n)$ = $f_{L'}(n)$. Furthermore, $L$ is called {\it strongly counting-regular} if, for any regular language $L_1$, $L \cap L_1$ is counting-regular.
Let $k \ge 1$.  
A language $L$ is $k$-slender if $ f_L(n) \le  k$ for all $n$, and $L$ is thin if is is $1$-slender. Furthermore, $L$ is slender if it is $k$-slender for some $k$.

A language $L\subseteq \Sigma^*$ is {\em bounded} if there exist (not necessarily distinct)
words $w_1, \ldots, w_k \in \Sigma^+$ such that $L \subseteq w_1^* \cdots w_k^*$;
it is also {\em letter-bounded} if each of $w_1, \ldots, w_k$ are letters. 

Let $\mathbb{N}$ be the set of positive integers and $\mathbb{N}_0 = \mathbb{N} \cup \{0\}$. 
A {\em linear set} is a set
$Q \subseteq \mathbb{N}_0^m$ if there exist $\vec{v_0},
\vec{v_1}, \ldots, \vec{v_n}$ such that
$Q = \{\vec{v_0} + i_1 \vec{v_1} + \cdots + i_n \vec{v_n} \mid
i_1, \ldots, i_n \in \mathbb{N}_0\}$. The vector
$\vec{v_0}$ is called the {\em constant}, and 
$\vec{v_1}, \ldots, \vec{v_n}$ are called the {\em periods}. We also say that $Q$ is the
linear set generated by constant $\vec{v_0}$ and periods $\vec{v_1}, \ldots, \vec{v_n}$.
A linear set is called {\em simple} if the periods form a basis.
A {\em semilinear set} is a finite union of linear sets.
And, a semilinear set is {\em semi-simple} if it is the finite disjoint union of simple sets
\cite{Flavio,Sakarovitch}.

A language $L \subseteq \Sigma^*$ is {\em semilinear} if $\psi(L)$ is a semilinear set.
Equivalently, a language $L$ is semilinear if and only if there is a regular language
$L'$ with the same Parikh map (they have the same commutative closure) \cite{harrison1978}.
The {\em length set} of a language $L$ is the set $\{n \mid w \in L, n = |w|\}$.
A language $L$ is {\em length-semilinear} if the {\em length set} is a semilinear set;
i.e.\ after mapping all letters of $L$ onto one letter, the Parikh map is semilinear, which is equivalent to it being regular.

A language $L \subseteq \Sigma^+$ is a {\em code} if $x_1 \cdots x_n = y_1 \cdots y_m, x_i, y_j \in L$,
implies $n=m$ and $x_i = y_j$, for $i$, $1 \leq i \leq n$. Also, $L$ is a prefix code if
$L \cap L \Sigma^+ = \emptyset$, and a suffix code if $L \cap \Sigma^+ L = \emptyset$. 
See \cite{CodesHandbook} for background on coding theory.

A language family $\LL$ is said to be {\em semilinear} if all $L \in \LL$ are semilinear. It is said that language family $\LL$ is a {\em trio}
if $\LL$ is closed under inverse homomorphism, $\epsilon$-free homomorphism, and intersection with regular languages. 
In addition, $\LL$ is a full trio if it is a trio closed under homomorphism; and a full AFL is a full trio closed
under union, concatenation, and Kleene-*.
Many well-known families form trios, such as each family of the Chomsky hierarchy \cite{Hopcroft}.
The theory of these types of families is explored in \cite{G75}.
When discussing a language family that has certain properties, such as a semilinear trio, we say that the family has {\em all properties effective} if all these properties provide effective constructions. For semilinearity, this means that 
there is an effective construction to construct the constant and periods from each linear set making up the semilinear set.

A pushdown automaton $M$ is $t$-reversal-bounded if $M$ makes at most
$t$ changes between non-decreasing and non-increasing the size of its pushdown on every input, and it is reversal-bounded if it is $t$-reversal-bounded for some $t$. A pushdown automaton $M$ is unambiguous if, for all $w\in \Sigma^*$, there is at most one accepting computation of $w$ by $M$. More generally, $M$ is $k$-ambiguous if there are at most $k$ accepting computations of $w$. 

Let $\DPDA$ ($\PDA$) denote the class of deterministic (nondeterministic) pushdown automata (and languages). We also use $\CFL = \PDA$, the family of context-free languages.

We make use of one particular family of languages, which we will only describe intuitively (see \cite{ibarra1978} for formal details). Consider a nondeterministic machine with a one-way input and $k$ pushdowns, where each pushdown only has a single symbol plus a bottom-of-stack marker. Essentially, each pushdown operates like a counter, where each counter contains some non-negative integer, and machines can add or subtract one, and test for emptiness or non-emptiness of each counter. 
When $k = 1$, we call these nondeterministic (and deterministic) one counter machines. 
Although such a machine with two counters has the same power as a Turing machine \cite{Hopcroft}, if the counters are restricted, then the machine can have positive decidability properties. Let $\NCM(k,t)$ be the family of $k$ counter machines where the counters are $t$-reversal-bounded counters, and let $\DCM(k,t)$ be the deterministic subset of these machines. Also, let $\NCM$ be $\bigcup_{k,t\geq 1} \NCM(k,t)$ and $\DCM$ be $\bigcup_{k,t \geq 1} \DCM(k,t)$.  The class of $\PDA$'s or $\DPDA$'s augmented with reversal-bounded counter machines is denoted by $\NPCM$ or $\DPCM$ \cite{ibarra1978}.
It is known that $\NCM$ has a decidable emptiness problem, and $\DCM$ also has a decidable containment problem, with both being closed under intersection \cite{ibarra1978}. Furthermore, $\NCM$ 
and $\NPCM$ are semilinear trios.

\section{Counting-Regular Languages}
\label{sec:countingregular}

Obviously every regular language is counting-regular, and so are many non-regular languages. For example, $L_{Sq}$ = $\{ w w \mid w \in \{a, b\}^*\}$ is counting-regular since $L'$ = $(a(a+b))^*$ has the same number of strings of length $n$ as $L_{Sq}$ for all $n$. It is a simple exercise to show that $L_{Sq}$ is actually strongly counting-regular.
It is also easy to exhibit languages that are not counting-regular, e.g., $L_{bal}$ = $\{ w \mid w$ is a balanced parentheses string$\}$ over alphabet $\{ [, ]\}$. The reason is that there is no regular language $L$ such that the number of strings of length $2n$ in $L$ is the $n$'th Catalan number, as will be seen from the characterization theorem stated below.

Our goal in this section is to explore general families of languages that are counting-regular. 
We briefly note the following:
\begin{theorem}
Any family $\LL$ that contains some non-length-semilinear language $L$ is not counting regular.
\end{theorem}
\begin{proof}
Given such an $L$, then examine the set of all $n$ with $f_L(n) >0$. But, as every regular language $R$
is length-semilinear, the set of all $n$ with $f_R(n) >0$ must be different.
\qed \end{proof}
Thus, it is immediate that e.g.\ any family that contains some non-semilinear unary language
must not be counting-regular. This includes such families as those accepted by checking stack
automata \cite{CheckingStack}, and many others. Of interest are exactly what families of length-semilinear languages
are counting-regular. We will investigate these questions here.

The following result is due to Berstel \cite{Berstel2}.

\begin{theorem}
\label{regularcharacterization}
Let $L$ be a regular language and let $f_L(n)$ denote the number of strings of length $n$. Then, one of the following holds:
\begin{enumerate}
\item[(i)] $f_L(n)$ is bounded by a constant $c$.

\item[(ii)] There is an integer $k > 0$ and a rational $c > 0$ such that $limsup_{n \rightarrow \infty} {{f_L(n)} \over {n^k}} = c$.

\item[(iii)] There exists an integer $k \geq 0$ and an algebraic number $\alpha$ such that \\
$limsup_{n \rightarrow \infty}  {{f_L(n)} \over {\alpha^n n^k}} =  c$ (where $c \neq 0$ is rational).
\end{enumerate}
\end{theorem}

We also need the following theorem due to Soittola \cite{Berstel}. We begin with the following definitions. 
A sequence $s$ = $\{s_n\}$, $n \geq 0$, is said to be the {\it merge} of the sequences $\{s^{(0)}\}, \ldots , \{s^{(p-1)}\}$, where $p$ is a positive integer, if 
$s_n^{(i)}$ = $s_{i+np}$ for $0 \leq i \leq p-1$. A sequence $\{s_n\}$ is said to be {\it regular} if there exists a regular language $L$ such that $f_L(n)$ = $s_n$ for all $n$.

Next we define a $\mathbb{Z}$-rational sequence as follows: A sequence $\{s_n\}$, $n \geq 0$, is $\mathbb{Z}$-rational if there is a matrix $M$ of order $d \times d$, a row vector $u$ of order $1 \times d$, and a column vector $v$ of order $d \times 1$ such that $s_n$ = $u \ M^n v$. All the entries in $M$, $u$ and $v$ are over $\mathbb{Z}$.
A $\mathbb{Z}$-rational sequence is said to have a {\it dominating pole} if its generating function $s(z)$ = $\sum_{n=0}^{\infty} s_n z^n$ 
can be written as a rational function $s(z)$ = $p(z)/q(z)$ where $p$ and $q$ are relatively prime polynomials, and $q$ has a simple root $r$ such that 
$r' > r$ for any other root $r'$.

Soittola's theorem \cite{Berstel} can be stated as follows.

\begin{theorem}
\label{Soittola}
A $\mathbb{Z}$-rational sequence  with non-negative terms is regular if and only if it is the merge of $\mathbb{Z}$-rational sequences with a 
dominating pole.
\end{theorem}

We will also need the following theorem due to B\'{e}al and Perrin \cite{Beal}.
\begin{theorem}
\label{Beal}
A sequence  $s$ is the generating sequence of a regular language over a $k$-letter alphabet if and only if both the sequences $s$ = $\{s_n\}$, $n \geq 0$ 
and $t$ = $\{(k^n - s_n)\}$, $n \geq 0$ are regular.
\end{theorem}

We now show the main result of this section. This result can be viewed as a strengthening
 of a result of Baron and Kuich \cite{Kuich} that $L(M)$ has a rational generating function if $M$ is an unambiguous finite-turn $\PDA$.

Here, an $\NTM$ is considered to have a one-way read-only input tape, plus a two-way read/write worktape that uses blank symbol $\blank$. Such a machine is said to be {\em reversal-bounded} if there is a bound on the number of changes in direction between moving left and right and vice versa on the worktape.
An $\NTM$ is said to be in {\em normal form} if,
whenever the worktape head moves left or right (or at the beginning of a computation) to
a cell $c$, then the next transition can change the worktape letter in cell $c$, but then the cell does not change again until after the worktape head moves.
This essentially means that if there is a sequence of `stay' transitions (that do not
move the read/write head) followed by a transition that moves, then only the first such transition can change the tape contents.
\begin{lemma}
Given an unambiguous reversal-bounded $\NTM$ $M$, there exists an unambiguous reversal-bounded
$\NTM$ $M'$ in normal form such that $L(M) = L(M')$.
\end{lemma}
\begin{proof}
Given $M$, an $\NTM$ $M'$ is constructed as follows:
After a transition that moves the read/write head (or at the first move of the computation),
instead of simulating a `stay' transition directly, $M'$ guesses the final value to be
written on the cell before the head moves, and writes it during the first `stay' transition.
$M'$ then continues to simulate the sequence of stay transitions without changing the
value in the cell but remembering it in the finite control. Then, during the last transition of this sequence
that moves the tape head, $M'$ verifies that it guessed the final value correctly.
Certainly, $M'$ is reversal-bounded if and only if $M$ is reversal-bounded. Further,
as $M$ is unambiguous, there is only one computation that is accepting on every word in $L(M)$. And, in $M'$ therefore, there can only be one value guessed on each sequence of `stay' transitions that leads to acceptance. Thus, $M'$ is unambiguous. 
\qed\end{proof}

It will be shown that every such $\NTM$ is counting regular.
\begin{theorem}
\label{main2}
Let $M$ be an unambiguous reversal-bounded $\NTM$ over a $k$ letter alphabet. Then $L(M)$ is strongly counting regular, where the regular language is over a $k+1$ letter alphabet. 
\end{theorem}
\begin{proof}
First we note that it is enough to show that $L(M)$ is counting regular, as the class of languages accepted by unambiguous reversal-bounded $\NTM$s are closed under
intersection with regular languages.

Let $M= (Q,\Sigma,\Gamma,\delta,q_0,F)$ be an unambiguous $\NTM$ that is $t$-reversal-bounded such that
$M$ is in normal form. Also, assume without loss of generality that $t$ is odd.

Intuitively, the construction resembles the construction that the store languages (the language
of all contents of the worktape that can appear in an accepting computation) of every such
$\NTM$ is a regular language \cite{StoreLanguages}. 
A $2\NFA \ M'$ is constructed that has $t+1$ ``tracks'', and it uses the first track for simulating $M$
before the first reversal, the second track for simulation of $M$ between the first and the second reversal, etc. Thus the input to $M'$ is the set of strings in which the first track contains an input string $x$ of $M$, and the other tracks are annotated with the contents of the read-write tape during moves of $M$ between successive reversals. Formal details are given below.

A $2\NFA$ $M'= (Q',\Sigma',\delta',q_0',F')$ is constructed as follows:
Let $C = [(Q \times (\Sigma \cup \{\epsilon\})\times Q \times \Gamma) \cup \{\blank\}]^{t+1}$
(the $t+1$ tracks; each track is either a blank, or some tuple in $Q \times (\Sigma \cup \{\epsilon\}) \times Q \times \Gamma$). Also, let $C_i$ have $t+1$ tracks where the $i$'th track, for $1 \leq i \leq t+1$, contains an element from $ Q \times (\Sigma \cup \{\epsilon\}) \times Q \times \Gamma$, and
all other tracks contain new symbol $\#$.
Let $\Sigma' = C \cup C_1 \cup \cdots \cup C_{t+1}$.
To simulate moves between the $(i-1)$st reversal and the $i$'th reversal, $M'$ will
examine track $i$ of a letter of $C$ to simulate the first transition after the tape head
moves to a different cell (or at the first step of a computation), and track $i$ of letters
of $C_i$ to simulate any stay transitions that occur before the tape head moves again, followed by the transition that moves.

Let $X = C_{t+1}^* \cdots C_4^* C_2^* C C_1^* C_3^* \cdots C_t^*$. Let $h_i$ be a homomorphism that maps each string in $(\Sigma')^*$ to the $i$'th track for symbols in $C \cup C_i$, and erases all symbols of $C_j, j \neq i$. Also, $\bar{h_i}$ is a homomorphism that maps each string in $(\Sigma')^*$
to the $i$'th track if it is not $\blank$, and $\epsilon$ otherwise, for symbols in $C \cup C_i$, and erases all symbols of $C_j$, $j \neq i$.
Then $M'$ does the following:
\begin{enumerate}
\item $M'$ verifies that the input $w$ is in $X^*$, and that no letter of  $[\blank]^{t+1}$ is used in $w$.
\item $M'$ verifies that for each $i$, $h_i(w) \in \blank^* (Q \times (\Sigma \cup \{\epsilon\}) \times Q \times \Gamma)^* \blank^*$, so blanks can only occur at the ends.
\item $M'$ verifies that $w$ represents an accepting computation of $M$ as follows: $M'$ goes to the first symbol of $C$ with a non-blank in the first track. Say 
$\bar{h_1}(w) = (p_1,a_1,p_1',d_1) \cdots (p_m,a_m,p_m',d_m), m \geq 1$. For $j$ from $1$ to $m$,
we say that $j$ is from $C$ if $(p_j,a_j,p_j',d_j)$ is from a symbol of $C$ and not $C_1$, and we say $j$ is from $C_1$ otherwise. It verifies from 
left-to-right on $w$ that for each $j$ from $1$ to $m$, there is a transition of $M$ that switches from $p_j$ to $p_j'$ while reading $a_j \in \Sigma \cup \{\epsilon\}$ as input on worktape letter $\blank$ if $j$ is in $C$, and $d_{j-1}$ otherwise, replacing it with $d_j$, that:
\begin{itemize}
\item moves right on the worktape, if $j<m$ and $j+1$ is from $C$,
\item `stay's on the worktape, if $j<m$ and $j+1$ is from $C_1$,
\item moves left on the worktape, if $j = m$.
\end{itemize}
$M'$ also verifies that $p_1 = q_0$, and that for each $1 \leq j <m$, $p_j' = p_{j+1}$, 
At the symbol of $w$ where $(p_m,a_m,p_m',d_m)$ occurs (at a point of reversal), $M'$ verifies
that all symbols of $C$ to the right in track 1 and track 2 are blanks as well as there being no symbols from
$C_1 \cup C_2$. Then $M'$ returns to the point of reversal, remembers $p_m'$ in the finite control,
and then returns to the rightmost symbol of $C$ in $h_1(w)$. From this symbol to the right, all
values in $C$ have blanks in the second track. It is verified that the second track has a non-blank
and $p_m'$ for the first state. Then on $\bar{h_2}(w)$ from right-to-left (using symbols of $C_2$ instead of
$C_1$ that occur to the left, read after a symbol from $C$), $M'$ continues the simulation in a similar
fashion until the second reversal, but $M'$ instead verifies that for all symbols of $C$ read, it represents a transition that rewrites the worktape symbol on the first track with the worktape symbol on the second track. 
Since $M$ only changes values in tape cells the first move after arriving at a cell (corresponding to letters
of $C$), after each reversal, $M'$ can ``lookup the value in the cell'' by using the previous track.
$M'$ continues this process until up to after reversal $t$.
Then $M'$ accepts if this process results in a final state.
\end{enumerate}

It is evident that $w$ describes  both the tape contents and also encodes the input for every accepting computation. Since each $x \in L(M)$ has exactly one accepting computation, and by the ordering of $X$,
there must be a unique $w \in L(M')$ with $x$ as the input word. Call this unique word $w$, $\code(x)$.  Furthermore, consider homomorphisms
$g_1,\ldots, g_{t+1}$ such that $g_i$ maps each symbol of $C \cup C_i$ to the letter of $\Sigma \cup \{\epsilon\}$
in track $i$, and erases all other symbols.
Given every word $w \in L(M')$, the word $g_1(w) g_2(w)^R \cdots g_t(w) g_{t+1}(w)^R$ gives
this word $x$. Hence, there is a bijection between
$x \in L(M)$ and $\code(x) \in L(M')$.

However, the length of $\code(x)$ can be different than that of $x$. Every letter in $C_1 \cup \cdots  \cup C_{t+1}$ encodes either one or zero letters of $\Sigma$ (depending on whether it contains a letter from $\Sigma$ or $\epsilon$).
For every letter of $C$, it encodes between $0$ and $t+1$ letters of $\Sigma$ (depending on the 
number of tracks encoding letters or $\epsilon$). Let $h$ be a homomorphism from
$(\Sigma')^*$ to $(\Sigma' \cup \{\$\})^*$ (where $\$$ is a new symbol) that acts as follows: $h$ fixes
all letters of $C_1 \cup \cdots \cup C_{t+1}$ that have letters of $\Sigma$, and erases those in $\epsilon$; it erases all symbols of $C$ where every non-blank track encodes a transition on $\epsilon$ (call these letters $C_{\epsilon}$), and it maps all other letters of $C$ (call the set of these $\bar{C}$) $c \in \bar{C}$ to $c \$^{l-1}$, where $c$ has $l$
tracks encoding transitions on $\Sigma$ ($l \geq 1$, otherwise $c$ would be in $C_{\epsilon}$).
Define $\overline{\code}(x) = h(\code(x))$.

To see that there is a bijection from $\code(x)$ to $\overline{\code}(x)$, first notice that $h$ is a function. In the other
direction, the number of $\$$'s is determined by the preceding letter of $C$; and even though letters
of $\Sigma'$ are erased, as $\overline{\code}(x)$ still encodes all letters of
$\Sigma$ read (i.e.\ $x = g_1(\overline{\code}(x)) g_2(\overline{\code}(x))^R \cdots g_{t+1}(\overline{\code}(x))^R$, where each $g_i$ is
extended to erase $\$$), and therefore this coding is still unique.

In addition, let $R= \{\overline{\code}(x) \mid x \in L(M)\}$ which is regular since languages accepted by 
$2\NFA$s are regular, and regular languages are closed under homomorphism.
Lastly, for each $x \in L(M)$, $|\overline{\code}(x)| = |x|$ by the reasoning above.

Surprisingly, we can reduce the size of the alphabet to $s+1$. More precisely, let $M$ be an unambiguous $\NTM$ as above over $\Sigma$ with $s = |\Sigma|$. We will show that there is a regular language $L_1$ over an alphabet of size $s+1$ such that $f_{L(M)}(n)$ = $f_{L_1}(n)$ for all $n$. As we have shown above, there is a regular language $L(M_1)$ such that $f_{L(M_1)}(n)$ = $f_{L(M)}(n)$ for all $n$. Thus, by Theorem \ref{Soittola},  the generating function $a(z)$ associated with the sequence $(f_{L(M)}(n))$, $n \geq 0$, is a rational function $p(z)/q(z)$ that satisfies the conditions of Theorem \ref{Soittola}. Let $b_n$ be defined as $b_n$ = $(k+1)^n - a_n$. The generating function $b(z)$ for the sequence $(b_n)$, $n \geq 0$, is ${p(z)} \over {(1-(k+1)z)\ q(z)}$. This is a rational function with dominating pole $1 \over{k+1}$ and hence it satisfies Theorem \ref{Soittola}. Since both $(a_n)$ and $(b_n)$ = $((k+1)^n - a_n)$ are regular sequences, by Theorem \ref{Beal}, there is a regular language $L_1$ over a $k+1$ letter alphabet such that $f_{L_1}(n)$ = $f_{L(M)}(n)$ for all $n$.
\qed
\end{proof}

We conjecture that the result also holds for $\NTM$s with finite-crossing (where there is a bound on the number of times the boundary between any two cells are crossed) worktapes.

Many different machine models can be simulated by unambiguous $\NTM$s with a one-way read-only
input tape and a reversal-bounded worktape. These include unambiguous reversal-bounded queue automata where
the store is a queue with a bound on the number of switches between enqueueing an dequeueing,
and also unambiguous reversal-bounded $k$-flip $\NPDA$s, which are like $\NPDA$s with the additional ability to flip the
stores up to $k$ times. As deterministic models are all unambiguous, this applies to deterministic models as well. We leave the details of the simulation to the reader.

\begin{corollary}
\label{cortomain}
Let $M$ be a machine accepted by any of the machine models below, with a one-way input:
\begin{itemize}
\item unambiguous nondeterministic reversal-bounded $\NPDA$s,
\item reversal-bounded $\DPDA$s,
\item unambiguous nondeterministic reversal-bounded queue automata,
\item deterministic reversal-bounded queue automata,
\item unambiguous nondeterministic reversal-bounded $k$-flip $\NPDA$s,
\item reversal-bounded $k$-flip $\DPDA$s,
\item $\DTM$s with a reversal-bounded worktape.
\end{itemize}
Then the language $L(M)$ is strongly counting-regular. 
\end{corollary}



\vskip 0.25cm
The next theorem shows that some natural extensions of the model in the above theorem accept some non-counting-regular languages. In the proofs below, we use the well-known fact  \cite{Salomaa} that if $L$ is regular, then $f_L( n)$ is rational.

\begin{theorem}
\label{counter-examples}
The following families of languages are not counting-regular:
\begin{enumerate}

\item languages accepted by deterministic $1$-counter machines (no reversal-bound), 

\item $\DCM(2,1)$ (deterministic $1$-reversal-bounded $2$-counter languages),

\item $\NCM(1,1)$ (nondeterministic $1$-reversal-bounded $1$-counter languages),

\item $2$-ambiguous $\NCM(1,1)$ languages.

\end{enumerate}

\end{theorem}
\begin{proof}
Consider the language $L_{Eq}$ over $\{a, b\}$ that accepts strings with an equal number of $a$'s and $b$'s. Clearly $L_{Eq}$ can be accepted by a deterministic $1$-counter machine (without restriction on the number of reversals). We will now show that there is no regular language $L$ such that $f_L(n)$ = $f_{L_{Eq}}(n)$ for all $n$. It is easy to see that $f_{L_{Eq}}(n)$ is 0 if $n$ is odd, and and is ${2n} \choose n$ if $n$ is even. But ${{2n} \choose n} \approx {{c 2^n} \over {\sqrt n}}$ for a constant $c$. Since the exponent of $n$ in the asymptotic expression for $f_{L_{Eq}}(n)$ is  $-{1 \over 2}$ which is not a positive integer, the conclusion follows from Theorem \ref{regularcharacterization}. 

The same language is a witness for the second class: it is easy to construct a $\DCM(2,1)$ machine that accepts $L_{Eq}$: increment counter one (two) for each $a$ ($b$) read and verify the counter values are the same at the end by decrementing simultaneously and accept if the counters reach the value 0 at the same time.

To show 3, we use a language from \cite{Flajolet}: Let $S$ = $\{ w \ |\ w = a^n b v_1 a^n v_2$ for some $v_1$ and $v_2$ in $\{a, b\}^*\}$. It is easy to see that $S$ can be accepted by an $\NCM(1,1)$. It is shown in 
\cite{Flajolet} that the generating function $S(z)$ for $S$ is $S(z)$ = ${z(1-z)} \over {1-2z}$ $\sum_{n \geq 1} {z^{2n} \over {1-2z+z^{n+1}}}$. Since $S(z)$ has countably many poles (one for each $n$ in the infinite sum), it follows that $S(z)$ is not algebraic. Since the generating function for any regular language is algebraic, it follows that $S$ is not counting-regular. 

 Finally, we will show 4. Note that the claim of 4 is stronger than 3, but the counter-examples we offer for 3 and 4 exhibit an interesting contrast. The former example has the property that its generating function has an infinite number of poles. The example we present now has only a finite number of poles but is not algebraic. Consider the languages $L_3$ = $\{ x 1^n \ | $ $x \in \{a, b\}^*, \ |x|_a = n \}$ and $L_4$ = $\{ x 1^n \ | $ $x \in \{a, b\}^*, \ |x|_b = n \}$. Clearly $L_3$ and $L_4$ can be accepted by a $\DCM(1,1)$ and hence are strongly counting-regular by Corollary \ref{cortomain}. (For $L_3$, the counter machine pushes a 1 on the counter on each $a$, skipping over $b$'s. When a first 1 is reached, it starts popping and makes sure that no $a$ or $b$ is seen again, and accepts when the counter reaches 0.) We will now show that $L_5$ = $L_3 \cup L_4$ is not counting-regular, by showing that its generating function is not rational. 

Let $L_6$ = $L_3 \cap L_4$. It is easy to see that $L_6$ =  $\{ x 1^n \ | $ $x \in \{a, b\}^*, \ |x|_a = |x|_b = n \}$. It can be checked that
$$f_{L_6}(n)  = \begin{cases}
{2n \choose n} & \mbox{~if $n \equiv 0$ (mod 3)}\\
0 & \mbox{else}.
\end{cases}
$$

So the generating function of $L_6$ is $f_{L_6}(z)$ = $1 \over{\sqrt {1-4z^3}}$. Clearly $f_{L_6}(z)$ is not rational. Suppose $L_5$ is counting-regular. Then, $f_{L_5}(z)$ is rational, and $f_{L_6}(z)$ = $f_{L_3}(z) + f_{L_4}(z) - f_{L_5}(z)$ is rational, a contradiction. 

Since $L_5$ can be accepted by a $2$-ambiguous $\NCM(1,1)$, the claim follows.
\qed \end{proof} 
 
Are there counting-regular languages that are not strongly counting-regular? Specifically,  since we showed that all unambiguous reversal-bounded $\PDA$'s are strongly counting-regular, is there an unambiguous $\PDA$ $L$ (that is not reversal-bounded) such that $L$ is counting-regular, but not strongly counting-regular? We discuss this issue next. 
In fact, we show that there is a deterministic $1$-counter language (also a $\DCM(2,1)$ language) that is counting-regular, but not strongly counting-regular, namely:
 $L_{\MAJ}$ = $\{ x \in (0+1)^* \ | \ $ $x$ has more 1's than 0's or has equal number of 0's and 1's and starts with 1$\}$. 

\begin{theorem}
$L_{\MAJ}$ is counting-regular but is not strongly counting-regular.
\label{MAJ}
\end{theorem}
\begin{proof}
Since the number of strings of length $n$ in $L_{\MAJ}$ is exactly $2^{n-1}$ for all $n\geq 1$, it is counting-regular.

Since $L_{\MAJ}$ is counting-regular, the smallest $\DFA$ $M$ such that $L_{\MAJ} \cap L(M)$ does not have a regular counting function must have at least two states. Somewhat surprisingly, such a $2$-state $\DFA$ exists. Consider the $\DFA$ $M_2$ in Figure \ref{fig3}.
\begin{figure}
\begin{center}
\includegraphics[width=2.2in]{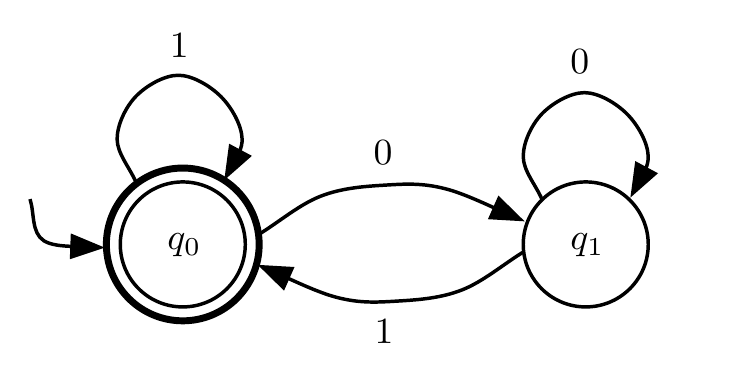}
\caption{$\DFA$ $M_2$}
\label{fig3}
\end{center}
\end{figure}
Let $L_2$ = $L_{\MAJ} \cap L(M_2)$. To show our claim, we need to establish that there is no regular language $L$ such that $f_{L_2}( n)$ = $f_L(n)$.  

We will obtain a closed-form expression for $f_{L_2}(n)$ and from it, we will obtain the generating function for the sequence $\{ f_{L_2}(n) \}$, $n$ = $0,\ 1, \ldots $  as follows. Note that the $\DFA$ $M$ has the property that all the transitions into a state are labeled by the same input. This means that there is a bijective mapping between an input string $w$ and the sequence of states visited on input $w$. Specifically, if we omit the first state (which is $q_0$, the start state), every input string of length $n$ can be bijectively mapped to the sequence of length $n$ of states visited. Thus, when $n$ is odd, the number of strings of length $n$ in $L_2$ is equal to the number of sequence of states of length $n$ that have more $q_0$ than $q_1$ and ending with $q_0$. (The requirement of ending with $q_0$ ensures that the string is in $L(M_2)$.) When $n$ is even, the number of strings of length $n$ in $L_2$ is equal to the number of sequences of states of length $n$ that have more $q_0$ than $q_1$ and ending with $q_0$, plus the number of sequences of states of length $n$ with equal number of $q_0$ and $q_1$ that begins and ends with $q_0$. Thus:
$$f_{L_2}(n)  = \begin{cases}
2^{n-2}+ {n-1 \choose {n-1 \over 2}} & \mbox{~if $n$ is odd and~} n \geq 3,\\
2^{n-2}+ {{n-2} \choose {(n-2) \over 2}} & \mbox{~if $n$ is even (and $n \geq 2$)}.
\end{cases}
$$
From this, we can explicitly obtain the generating function of $L_2$ as:
$1 + {z^2 \over {1-2z} } + {{z^2+z}  \over {\sqrt {1-4z^2}}}$. This function is not a rational function and the conclusion follows.
\qed \end{proof}

Finally, we will show an interesting language that is both context-free and strongly counting-regular. This language was presented in \cite{Casti} as a member of the class $\RCM$ --- a class defined as follows: $\RCM$ is the class of the languages given by $< R,C, \mu>$ where $R$ is a regular language, $C$ a system of linear constraints and $\mu$, a length-preserving morphism that is injective on $R \cap [C]$. The specific language we denote by $L_{\RCM}$
 is defined as follows:
$$L_{\RCM} = \{w \in \{a,b\}^* \ |\ w[|w|_a] = b\}.$$
(From the definition, it is clear that $a^n$ is not in $L_{\RCM}$.  But the definition does not specify the status of the string $b^n$ since there is no position 0 in the string. We resolve this by explicitly declaring the string $b^n$ to be in $L_{\RCM}$ for all $n$.) 

We first observe that $L_{\RCM}$ is in $\NCM(1,1)$ and hence is context-free: We informally describe a nondeterministic 1-reversal-bounded 1-counter machine $M$ for $L_{\RCM}$. Let $w$ = $w_1 \cdots w_n \in L_{\RCM}$ and suppose $|w|_a$ = $i$ and
$|w_1 w_2 \cdots w_{i-1}|_a$ = $j$. Then, it is clear that $w_i$ = $b$, $j < i$ and $|w_{i+1} \cdots w_n|_a$ = $i-j$. Now we describe the operation of $M$ on input string $w$ (not necessarily in $L_{\RCM}$) as follows. $M$ increments the counter for every $b$ until it reaches position $i$. $M$ guesses that it has reached position $i$, and checks that the symbol currently scanned is $b$, then switches to decrementing phase in which it decrements the counter for each $a$ seen. When the input head falls off the input tape, if the counter becomes 0 it accepts the string. If $w \in L_{\RCM}$, it is clear that when $M$ switches from incrementing to decrementing phase, the counter value is $i-j$ and hence when it finishes reading the input, the counter will become 0. The converse is also true. Thus it is clear that $L(M)$ = $L_{\RCM}$.

An interesting aspect of the next result is that it applies the simulation technique of Theorem \ref{main2} twice --- by first bijectively (and in length-preserving way) mapping $L_{\RCM}$ to a language accepted by $\DCM(1,1)$, then applying Theorem \ref{main2} to map it to a regular language. 

\begin{theorem}
$L_{\RCM}$ is strongly counting-regular.
\end{theorem}
\begin{proof} Consider $L'$ = $\{ [a_1, b_1] [a_2, b_2] \cdots [a_n, b_n]\ |\ a_1 a_2 \cdots a_n \in L_{\RCM}$, and $b_i$ = $c$ for all $1 \leq i \leq k$, and  $b_i$ = $d$ for all $k+1 \leq i \leq n$ where $k$ is the number of $a$'s in $a_1 a_2 \cdots a_n\}$. Thus, $L'$ is defined over $\Sigma_1$ = $\{[a,c], [b,c], [a,d],[b,d]\}$. As an example, the string  $[a, c][b,c][a, d]$ is in $L'$ since $aba \in  L_{\RCM}$. It is easy to see that there is a bijective mapping between strings in $L_{\RCM}$ and strings in $L'$. Given a string $w \in L_{\RCM}$, there is a unique string $w'$ in $L'$ whose upper-track consists of $w$ and the lower-track consists of $c^k d^{n-k}$ where $|w|$ = $n$ and $|w|_a$ = $k$. Since $k$ and $n$ are uniquely specified for a given string $w$, the string $w'$ is uniquely defined for a given $w$. Conversely, the projection of a $w' \in L'$ to its upper-track string uniquely defines a string in $L_{\RCM}$. 

Let $R$ be a regular language over $\{a, b\}$. Define a language $R' = \{[a_1, b_1] [a_2, b_2] \cdots [a_n, b_n]\ |\ a_1 \cdots a_n \in L' \cap R\}$. It is clear that there is a 1-1 correspondence between strings in $R'$ and $L_{RCM} \cap R$. Finally, we note that there is a $\DCM(1,1)$ $M$ that accepts $R'$: $M$ simulates the $\DFA$ for $R$ on the upper-track input and at the same time, performs the following operation. For every input $[b,c]$, the counter is incremented, and for every $[a,d]$ the counter is decremented. It also remembers the previous symbol scanned, and checks that the previous symbol scanned before the first occurrence of $[b,d]$ is $[a,c]$. Finally, when the counter value becomes 0, the string is accepted. Note that since all $d$'s in the second track occur after the $c$'s, the counter reverses at most once.

Since all $\DCM(1,1)$'s are $1$-reversal-bounded $\DPDA$'s, then by Corollary \ref{cortomain}, $R'$ is counting-regular. Hence, $L_{\RCM} \cap R$ is counting-regular.
\qed
\end{proof}

In this section, we have shown that the languages accepted by unambiguous nondeterministic Turing machines with a one-way read-only input tape and a reversal-bounded worktape are strongly counting-regular. We also showed that some natural extensions of this class fail to be counting-regular. We presented some relationships between counting-regular languages and the class $\RCM$. However, our understanding of which languages in $\DCM$, $\NCM$, or $\CFL$ are (strongly) counting-regular is quite limited at this time.

\section{Bounded Semilinear Trio Languages are Counting-Regular}
\label{bounded}

In this section, we will show that all bounded languages in any semilinear
trio are counting-regular.

First, the following is known \cite{CIAA2016} (follows from results in \cite{ibarra1978}, and the fact that all bounded $\NCM$
languages are in $\DCM$ \cite{IbarraSeki}):

\begin{lemma}
Let $u_1, \ldots, u_k \in \Sigma^+$, and let $\phi$ be a function from $\mathbb{N}_0^k$ to $u_1^* \cdots u_k^*$ which associates
to every vector $(l_1, \ldots, l_k)$, the word $\phi((l_1, \ldots, l_k)) = u_1^{l_1} \cdots u_k^{l_k}$.
Then the following are true:
\begin{itemize}
\item given a semilinear set $Q$, then $\phi(Q) \in \DCM$,
\item given an $\NCM$ (or $\DCM$) language $L \subseteq u_1^* \cdots u_k^*$, then $IND(L) = \{(l_1, \ldots, l_k) \mid \phi((l_1, \ldots, l_k)) \in L\}$ is a semilinear set.
\end{itemize} Moreover, both are effective.
\label{phi}
\end{lemma}

Recall that two languages $L_1$ and $L_2$ are called commutatively equivalent if there is a Parikh-map preserving bijection between them. Therefore, a language $L$ being commutatively equivalent to some regular language is stronger than saying it is counting-regular.
Split across three papers in \cite{FlavioBoundedSemilinearpaper1,FlavioBoundedSemilinearpaper2,FlavioBoundedSemilinear}, it was shown that all bounded semilinear languages ---
which are all bounded languages where $IND(L)$ is a semilinear set --- are commutatively
equivalent to some regular language, and are therefore counting-regular. Recently, it was shown that
all bounded languages from any semilinear trio are in $\DCM$ \cite{CIAA2016} (and are therefore
bounded semilinear by Lemma \ref{phi}). This enables us to conclude that all bounded languages
in any semilinear trio are commutatively equivalent to some regular language, and thus counting-regular.
However, as the proof that all bounded  semilinear languages are commutatively equivalent to some regular language is quite lengthy, we provide a simple alternate proof that all bounded languages in any semilinear trio (all bounded semilinear languages) are counting-regular. The class $\DCM$ plays a key role in this proof.

\begin{lemma}
Let $u_1, \ldots, u_k \in \Sigma^+$, $L\subseteq u_1^* \cdots u_k^*$ be a bounded $\DCM$ language, and let $\phi$ be a function from Lemma \ref{phi}.
There exists a semilinear set $B$ such that $\phi(B) = L$ and $\phi$ is injective on $B$. Also, the construction of $B$ is effective.
\label{injective}
\end{lemma}
\begin{proof}
Let $A = \{a_1, \ldots, a_k\}$, and consider the homomorphism $h$ that maps $a_i$ to $u_i$, for $1 \leq i \leq k$. It is known that
there exists a regular subset $R$ of $a_1^* \cdots a_k^*$ that  $h$ maps bijectively from $R$ onto $u_1^* \cdots u_k^*$ 
(the Cross-Section Theorem of Eilenberg \cite{Flavio}).
Let $L' = h^{-1}(L) \cap R$. Then $L'$ is in $\DCM$ since $\DCM$ is closed under inverse homomorphism and intersection with regular languages \cite{ibarra1978}. Hence, there is a semilinear set $B = IND(L')$ from Lemma \ref{phi}(2).
Then $\phi(B) = L$ since, given $(l_1, \ldots, l_k) \in B$, then
$u_1^{l_1} \cdots u_k^{l_k} \in L$, and given $w \in L$, by the bijection $h$,
there exists a string $a_1^{l_1} \cdots a_k^{l_k}$ of $R$ such that $a_1^{l_1} \cdots a_k^{l_k} = h^{-1}(w)$,
and so $w \in \phi(B)$. Also, $\phi$ is injective on $B$, as given two distinct elements
$(l_1, \ldots, l_k)$ and $(j_1, \ldots, j_k)$ in $B$, then both
$a_1^{l_1} \cdots a_k^{l_k}$ and $a_1^{j_1} \cdots a_k^{j_k}$ are in $R$, which means that $h$ maps them onto
different words in $u_1^* \cdots u_k^*$ since $h$ is a bijection.
\qed \end{proof}


The proof of the next result uses similar techniques as the proof that all bounded context-free languages are counting-regular from \cite{Flavio}. But because there are key
differences, we include a full proof for completeness.

\begin{lemma}
Let $L \subseteq u_1^* \cdots u_k^*$ be a bounded $\DCM$ language
for given words $u_1, \ldots, u_k$. Then there exists an effectively constructible bounded
regular language $L'$ such that, for every $n \geq 0$, $f_L(n) = f_{L'}(n)$.
\end{lemma}
\begin{proof}

Let $\phi$ be a function from $\mathbb{N}_0^k$ to $u_1^* \cdots u_k^*$ such that 
$\phi((l_1, \ldots, l_k)) = u_1^{l_1} \cdots u_k^{l_k}$. By Lemma \ref{phi}, there exists a semilinear
set $B$ of $\mathbb{N}_0^k$ such that $\phi(B) = L$ \cite{ibarra1978}. Let
$B = B_1\cup \cdots \cup B_m$, where $B_i$, $1 \leq i \leq m$ are linear sets. 
Let $L_1 = \phi(B_1) \in \DCM$ (by Lemma \ref{phi}), 
$L_2 = \phi(B_2) - L_1 \in \DCM$ (by Lemma \ref{phi} and since $\DCM$ is closed under intersection
and complement \cite{ibarra1978}), etc.\ until $L_m = \phi(B_m) - (L_1 \cup \cdots \cup L_{m-1}) \in \DCM$ (inductively, by Lemma \ref{phi}, by closure of $\DCM$ under intersection, complement, and union). Then $L_1 \cup \cdots \cup L_m = L$, and also $L_1, \ldots, L_m$
are pairwise disjoint, and therefore
by Lemma \ref{injective}, there is a semilinear
set $B_i'$ such that $\phi(B_i') =L_i$, and $\phi$ is injective on $B_i'$.
It is known that, given any set of constants and periods generating a semilinear set $Q$, 
there is a procedure to effectively construct another set of constants and periods that forms a semi-simple set,
also generating $Q$ \cite{Flavio,Sakarovitch} (this is phrased more generally in both works, to say that the rational sets of a commutative monoid
are semi-simple; but in our special case it amounts to constructing a new set of constants and periods generating the same
semilinear set such that the linear sets are disjoint, and the periods generating each linear set form a basis). 
Hence, each $B_i'$ must also be semi-simple as well (generated by a possibly different set of constants
and periods). Let $B' = B_1' \cup \cdots \cup B_m'$. Since each word in $L$ is only in exactly one language of $L_1, \ldots, L_m$,
it follows that for each $(l_1, \ldots, l_k)$, $\phi((l_1, \ldots, l_k))$ is in at most one language of $L_1, \ldots, L_m$.
And, since $\phi$ is injective on each $B_i'$, it 
therefore follows that $\phi$ is injective on $B'$. Also, $\phi(B') = L = \phi(B)$.

The rest of the proof then continues just as Theorem 10 of \cite{Flavio} (starting at the second paragraph) which we describe. That is, define an alphabet $A = \{a_1, \ldots, a_k\}$, and let $\psi$ be the Parikh map of $A^*$ to $\mathbb{N}_0^k$.
For every linear set $B''$ making up any of the semilinear sets of some $B_i'$, let
$B''$ have constant $b_0$ and 
periods $b_1, \ldots, b_t$. Define the regular language $R_{B''} = v_0 v_1^* \cdots v_t^*$
where $v_0, \ldots, v_t$ are any fixed words of $A^*$ such that, for every $i$,
$\psi(v_i) = b_i$. Thus, $\psi(R_{B''}) = B''$, for each $B''$. Let $R$ be the union of all $C_{B''}$ over all the linear
sets making up $B_1', \ldots, B_m'$. Certainly $R$ is a regular language.

It is required to show that $\psi$ is injective on $R$. Indeed, consider
$x, y$ be two distinct elements in $R$. If $x,y$ are constructed from two different linear sets
from distinct semilinear sets $B_i', B_j', i \neq j$, then $\psi(x) \neq \psi(y)$ since the semilinear
sets $B_i'$ and $B_j'$ are disjoint. If $x,y$ are constructed from two different linear sets
making up the same semilinear set $B_i'$, then since $B_i'$ is semi-simple,
the linear sets must be disjoint, and hence $\psi(x) \neq \psi(y)$.
If $x,y$ are in the same linear set, then $\psi(x)\neq \psi(y)$ since the linear
set must be simple, its periods form a basis, and therefore, there is only one
linear combination giving each. Hence, $\psi$ is injective on $R$.

Consider the map such that, for every $i$, $1 \leq i \leq k$, $a_i$ maps to
$a_i^{|u_i|}$, and extend this to a homomorphism $\chi$ from $A^*$ to $A^*$.
Since $\chi(A)$ is a code, $\chi$ is an injective homomorphism of $A^*$ to
itself. Let $L' = \chi(R)$. Then $L'$ is a regular language.

Next, it will be shown that $f_L(n) = f_{L'}(n)$. Consider the relation
$\zeta = \phi^{-1} \psi^{-1} \chi$. Then, when restricting $\zeta$ to 
$L$, this is a bijection between $L$ and $L'$, since $\phi$ is a bijection
from $B'$ to $L$, $\psi$ is a bijection of $R$ to $B'$, and $\chi$ is a bijection
of $R$ to $L'$. It only remains to show that, for each $u \in L$,
\begin{equation}
|u| = |\zeta(u)|,
\label{conc}
\end{equation}
which therefore would imply $f_L(n) = f_{L'}(n)$.
For each $u \in L$, then
$u = u_1^{l_1} \cdots u_k^{l_k} = \phi((l_1, \ldots, l_k)) = \phi(\psi(x))$,
where $x$ is in $R$ and $\psi^{-1}((l_1, \ldots, l_k))$.
Since $|x| = \sum_{1 \leq i \leq k}|x|_{a_i} = \sum_{1 \leq i \leq k} l_i$,
then $$|\chi(x)| = \sum_{1 \leq i \leq k} |x|_{a_i} |\chi(a_i)| = \sum_{1 \leq i \leq k} l_i|\chi(a_i)| = \sum_{1 \leq i \leq k} l_i|u_i| = |u|.$$ Thus, \ref{conc} is true, and
the theorem follows.
\qed \end{proof}
We should note that if the bounded language $L \subseteq a_1^* \cdots a_k^*$, where
$a_1, \ldots, a_k$ are distinct symbols, then there is a simpler proof of the theorem
above as follows: Let $L \subseteq a_1^* \cdots a_k^*$ where the symbols are distinct.
Then the Parikh map of $L$ is semilinear, and therefore a regular language $L'$
can be built with the same Parikh map, and in this language $f_L(n) = f_{L'}(n)$. 
But when the bounded language $L$ is not of this form, this simpler proof does not work. 

Our next result is a generalization of the previous result.

\begin{theorem} 
Let $L \subseteq u_1^* \cdots u_k^*$, 
for words $u_1, \ldots, u_k$ where $L$ is in any  
semilinear trio $\LL$. 
There exists a bounded
regular language $L'$ such that, for every $n \geq 0$, $f_L(n) = f_{L'}(n)$.
Moreover, $L$ is strongly counting-regular. Furthermore, if $u_1, \ldots, u_k$ are given,
and all closure properties are effective in $\LL$, then $L'$ is effectively constructible.
\label{cor9}
\end{theorem}
Again, this follows from \cite{CIAA2016} since it is known that every
bounded language from any such semilinear trio where the closure properties are effective
can be effectively converted
into a $\DCM$ language. Strong counting-regularity follows since intersecting
a bounded language in a trio with a regular language produces another bounded language
that is in $\LL$, since trios are closed under intersection with regular languages.

Also, since the family of regular languages is the smallest semilinear trio \cite{G75}, it follows that the counting functions for the bounded languages in every semilinear trio are identical.
\begin{corollary}
Let $\LL$ be any semilinear trio. The counting functions for the bounded
languages in $\LL$ are identical to the counting functions for the bounded regular languages.
\end{corollary}

This works for many semilinear full trios. We will briefly discuss some in the next example.
\begin{example}
\label{semilinearfulltrioexamples}
The families accepted/generated from the following grammar/machine models form semilinear full trios (the closure properties and semilinearity are effective):
\begin{enumerate}
\item the context-free languages, $\CFL$s, 
\item one-way nondeterministic reversal-bounded multicounter machines, $\NCM$s, \cite{ibarra1978},
\item finite-index $\ETOL$ systems ($\ETOL$ systems where the number of non-active
symbols in each derivation is bounded by a constant) \cite{RozenbergFiniteIndexETOL}.
\item $k$-flip $\NPDA$s ($\NPDA$s with the ability to ``flip'' their pushdown up to $k$ times) \cite{Holzer2003}. 
\item one-way reversal-bounded queue automata (queue automata with a bound on the number of switches between enqueueing and dequeueing) \cite{Harju}.
\item $\NTM$s with a one-way read-only input tape and a finite-crossing worktape \cite{Harju},
\item uncontrolled finite-index indexed grammars (a restricted version of indexed grammars, where
every accepting derivation has a bounded number of nonterminals),  \cite{LATA2017}.
\item multi-push-down machines (a machine with multiple pushdowns where the machine can simultaneously push to all pushdowns, but can only
pop from the first non-empty pushdown) \cite{multipushdown}.
\end{enumerate}
Moreover, all of these machine models can be augmented by reversal-bounded counters and
the resulting machines are semilinear full trios \cite{Harju,fullaflcounters}.
\end{example}

\begin{corollary}
Let $L \subseteq u_1^* \cdots u_k^*$, be a bounded language
for given words $u_1, \ldots, u_k$, such that $L$ is from any of the 
families listed in Example \ref{semilinearfulltrioexamples}.
Then there exists an effectively constructible bounded
regular language $L'$ such that, for every $n \geq 0$, $f_L(n) = f_{L'}(n)$.
\end{corollary}
Note that it is not assumed for these models that the machines are unambiguous, like in
Theorem \ref{main2}.

The results in this section assumed that the words $u_1, \ldots, u_k$ 
such that $L \subseteq u_1^* \cdots u_k^*$ are given. However, it is an open problem whether, given a language $L$ in an arbitrary semilinear trio $\LL$, it is possible to determine whether
$L \subseteq u_1^* \cdots u_k^*$ for some words $u_1, \ldots, u_k$.

\section{Closure Properties for Counting-Regular Languages}
\label{sec:closure}

In this section, we will address the closure properties of counting-regular languages, and also
counting-regular $\CFL$'s. 

First, it is immediate that counting-regular languages are closed under reversal (and since the $\CFL$s
are closed under reversal, so are the counting-regular $\CFL$s). Next Kleene-* will be addressed.
\begin{theorem}
\label{code}
If $L$ is counting-regular and $L$ is a code, then $L^*$ is counting-regular.
\end{theorem}
\begin{proof}
Since $L$ is a code, for each word $w \in L^*$, there is a unique decomposition of
$w = u_1 \cdots u_k$, where each $u_i \in L$. Since $L$ is counting-regular, there is
some regular language $R$ with the same counting function. From $R$, make
$R'$ where the first letter of each word is tagged with a prime, and all other letters are unmarked.
Now, $R'$ is a code because of the tagged letters, and $R'$ has the same counting function
as $R$.

Moreover, $(R')^*$ has the same counting function as $L^*$. Indeed, let $n \geq 0$. Consider all sequences $u_1, \ldots, u_k$ such that $n = |u_1| + \cdots + |u_k|$. Then for each
$u_i$, $L$ has the same number of words of length $|u_i|$ as does $R'$. Since $R'$ is a code,
it follows that there are the same number of such sequences using elements from $R'$.
\qed \end{proof}

A similar relationship to codes exists for concatenation.
\begin{theorem}
\label{prefixcode}
If $L_1, L_2$ are counting-regular and either $L_1$ is a prefix code or $L_2$ is a suffix code, then $L_1 L_2$ is counting-regular.
\end{theorem}
\begin{proof}
Assume first that $L_1$ is a prefix code, so that $L \cap L \Sigma^+ = \emptyset$. Let $w = uv, u \in L_1, v \in L_2$. Then this decomposition is unique since $L_1$ is a prefix code. Let $R_1, R_2$ be regular
languages with the same counting functions as $L_1, L_2$ respectively. Let $R_1'$ be obtained from $R_1$ by tagging the last letter with a prime. Then $R_1'$ is also a prefix code. Further,
the counting function for $R_1'R_2$ is equal to that of $L_1L_2$. 

The case is similar for $L_2$ being a suffix code.
\qed \end{proof}

\begin{corollary}
\label{corcode}
If $L_1, L_2 \subseteq \Sigma^*$ are counting-regular, and $\$,\#$ are new symbols, then $L^R,
L_1 \$L_2, \$L_1 \cup \#L_2$, and $(L\$)^*$ are counting-regular.
\end{corollary}
\begin{proof}
The first was discussed above. The second follows from Theorem \ref{prefixcode} and since
$L_1 \$$ is a prefix code. The fourth
follows from Theorem \ref{code} and since $L\$$ is a code. For the third, since $L_1,L_2$ are
counting-regular, this implies there exist regular languages $R_1, R_2$ with the same
counting functions, and $\$L_1 \cup \#L_2$ has the same counting function as 
$\$R_1 \cup \#R_2$.
\qed
\end{proof}

This means that even though e.g.\ non-reversal-bounded $\DPDA$s can accept non-counting-regular
languages but reversal-bounded $\DPDA$s cannot, if a $\DPDA$ was reversal-bounded but
reading a $\$$ caused a ``reset'' where the pushdown emptied, and another reversal-bounded computation was then possible, then this model would only accept counting-regular languages.
This is also the case with say $\DTM$s where the worktape was reversal-bounded, but reading
a $\$$ caused a reset, where more reversal-bounded computations were again possible.
This is quite a general model for which this property holds.

The next questions addressed are whether these are true when removing the $\$$ and $\#$
(or removing or weakening the coding properties).

\begin{theorem}
The counting-regular languages (and the counting-regular $\CFL$'s) are not closed under union or intersection with regular languages.
\end{theorem}
\begin{proof} Recall the $\DFA$ $M_2$ presented in Figure \ref{fig3}. Since $L_{\MAJ}$ is a counting-regular $\CFL$, and $L_{MAJ} \cap L(M_2)$ is not, the non-closure under intersection with regular sets follows.

For non-closure under union with regular sets, we show that  $L_{\MAJ} \cup L(M_2)$ is not counting-regular by explicitly computing the generating function for this language using the fact that there is a 1-1 mapping between strings of $L_{\MAJ}$ and $\overline{L_{MAJ}}$ and between strings of $L(M_2)$ and $\overline{L(M_2)}$.  Thus the generating functions for both $f_{L_{MAJ}}(n)$ and $f_{L(M_2)}(n)$ are ${1-z} \over {1-2z}$, from which it follows that the generating function for $L_{MAJ} \cup L(M_2)$ is ${{1-2z-z^2} \over {1-2z}}-$ ${{z^2+z}  \over {\sqrt {1-4z^2}}}$. Since this is not a rational function, the claim follows.
\qed \end{proof}
Thus, Corollary \ref{corcode} cannot be weakened to remove the marking from the marked union.

It is an open question as to whether Theorem \ref{code} can be weakened to remove the code
assumption, but we conjecture that it cannot.
However, for concatenation, we are able to show the following:
\begin{theorem}
The counting-regular languages (and counting-regular $\CFL$'s) are not closed under concatenation with regular languages.
\end{theorem}
\begin{proof}
Let $S_1$ = $\{ w \ |\ w = a^n b v a^n$ for some $v \in \{a, b\}^*\}$ and let 
$S$ = $\{ w \ |\ w = a^n b v_1 a^n v_2$ for some $v_1, v_2 \in \{a, b\}^*\}$ (as in Theorem \ref{counter-examples}).
It is easy to see that $S$ = $S_1 (a+b)^*$. We already showed that $S$ is not counting-regular. It is easy to show that $S_1$ is a counting-regular $\CFL$. In fact, we can explicitly exhibit $f_{S_1}(n)$ as follows: $f_{S_1}(0)$ = 0 and, for $n \geq 1$, $f_{S_1}(n)$ = $\sum_{j=0}^{\lfloor (n-1)/2 \rfloor} 2^{n-2i-1}$. From this, it is easy to see that $f_{S_1}(n)$ = $f_L(n)$ for the regular language $L$ with regular expression $(aa)^*b(a+b)^*$. 
\qed \end{proof}
Hence, Corollary \ref{corcode} and Theorem \ref{prefixcode} cannot be weakened to remove the marking with marked concatenation or the coding properties. It is an open problem as to whether
Theorem \ref{prefixcode} is true when $L_1$ or $L_2$ are codes (the set of
suffix codes together with the set of prefix codes is a strict subset of the set of codes \cite{CodesHandbook}).

\begin{theorem}
\label{rightquotient}
The counting-regular languages (and counting-regular $\CFL$'s) are not closed under right quotient with a single symbol, and are not closed under left quotient with a symbol.
\end{theorem}
\begin{proof} (sketch) 
First, it will be shown for right quotient.
The language $L_{\MAJ}$ used in Theorem \ref{MAJ} is counting-regular. In fact, it has exactly $2^{n-1}$ strings of length $n$ for all $n \geq 1$. We will outline an argument that $L$ =  $L_{\MAJ}\{0\}^{-1}$ is not counting-regular. The number of strings of length $n$ in $L$ is 
$2^{n-1} - {n \choose {n-1 \over 2}}$ (for odd $n$), $2^{n-1} - {n-2 \choose {n \over 2}}$ (for even $n$). Using a technique similar to the proof of Theorem \ref{MAJ}, we can show that the generating function for $L$ is not rational.

Next, it will be shown for left quotient.
Since $L_{\MAJ}$ is counting regular, and the reversal of every counting
regular language has the same counting function, then $L_{\MAJ}^R$ is also counting-regular. 
But, as in in proof for right quotient, $\{0\}^{-1}L_{\MAJ}^R$ is not 
counting-regular. 
\qed \end{proof}
The language $L_{MAJ}$ used is a deterministic one counter language with no reversal bound.
The following however provides a contrast, as it follows from closure properties of deterministic
machines.
When the languages are accepted by reversal-bounded machines, we have: 
\begin{theorem} 
\begin{enumerate}
\item If $L$ is accepted by a reversal-bounded 
$\DTM$, and $R$ is a regular language, then 
$LR^{-1}$ is counting-regular. 
\item If $L$ is accepted by an unambiguous reversal-bounded $\NTM$, and 
$x$ is a string, then $L \{x\}^{-1}$ is also counting-regular.
\end{enumerate}
\end{theorem}
\begin{proof} Part 1 follows from the fact that the languages accepted by 
reversal-bounded $\DTM$'s are closed under 
right-quotient with regular languages \cite{StoreLanguages}. 

For Part 2, clearly, if  $L$ is accepted by an unambiguous reversal-bounded 
 $\NTM$ $M$, we can construct 
an unambiguous reversal-bounded 
$\NTM$ $M$ accepting $L\{x\}^{-1}$. 
\qed \end{proof}
This result also holds for all machine models in Corollary \ref{cortomain}
(that is, all deterministic models listed there work with right quotient with
regular languages \cite{StoreLanguages}, and the unambiguous nondeterministic models there work
with right quotient with a word). It is an open question as to whether
unambiguous nondeterministic $\NTM$s with a reversal-bounded worktape
are closed under right quotient with regular languages, which would allow
part 2 to be strengthened.

Part 1 of of the next  theorem contrasts Part 1 of the previous theorem. 
\begin{theorem}
\begin{enumerate}
\item There is a counting-regular language $L$ accepted by a $\DCM(1,1)$ and distinct symbols $\$$ and $\#$ such that $\{\$,\#\}^{-1}L$ is not counting-regular. 
\item If $L$ is accepted by a reversal-bounded $\DPDA$ (resp., reversal-bounded 
unambiguous $\NPDA$, reversal-bounded  unambiguous $\NTM$), 
and $x$ is a string, then $\{x\}^{-1} L$ is also counting-regular. 
\end{enumerate}
\end{theorem}
\begin{proof}
For Part 1,  let 
$L_1 = \{x 1 ^n  \mid  x \in (a+b)^+, n = |x|_a \}$ and
$L_2 = \{x 1 ^n  \mid  x \in (a+b)^+, n = |x|_b \}$. 
Then $L_1$ and $L_2$ can each be accepted by a $\DCM(1,1)$. Let $L = L_1 \cup L_2$, shown
in Theorem \ref{counter-examples} to not be counting-regular.
Let $L' = \$L_1 \cup \#L_2$, which can also be accepted by a $\DCM(1,1)$, hence it is counting-regular.  However, $\{\$,\#\}^{-1}L' = L$ is not counting-regular. 
Part 2 is obvious.
\qed \end{proof}

It may seem obvious that for any counting-regular language $L$, $\overline{L}$ is counting-regular because of the following putative reasoning: If there is a regular language $L'$ whose counting function equals that of the counting function of $L$, the complement of $L'$ (which is regular) has the same counting function as that of $\overline{L}$. The fallacy in this argument is as follows. Suppose that the size of the alphabet over which $L$ is defined is $k$. The size of alphabet $k_1$ over which $L'$ is defined may be larger, i.e., $k_1 > k$. Thus, the complement of $L'$ has a counting function ${k_1}^n - f_{L'}(n)$ which is not the same as the counting function $k^n-f_L(n)$ of $\overline{L}$. 


In fact, the following result shows that the exact opposite of the fallacy is actually true.
\begin{theorem}
There is a counting-regular language $L$ (that is in $\P$, i.e., $L$ is deterministic polynomial time computable) such that $\overline{L}$ is not counting-regular.
\end{theorem}
\begin{proof} The proof relies on a result presented in B\'{e}al and Perrin \cite{Beal}. B\'{e}al and Perrin \cite{Beal} provide an example of a sequence $\{r_n\}$, $n \geq 0$ and an integer $k$ such that $\{r_n\}$ is not a counting function of any regular language, but $\{k^n - r_n\}$ is the counting function of a regular language. Specifically, it is shown in \cite{Beal} that the sequence $r_n$ = $b^{2n} \cos^2(n\theta)$, with $cos\ \theta$ = $a \over b$, where the integers $a$, $b$ are such that $b \neq 2a$ and $0 < a < b$, and $k$ such that $b^2 < k$, satisfies the properties stated above.

Define a language $L$ over an alphabet of size $k$ as follows: arrange the strings of length $n$ over the alphabet $\{0, 1, \ldots , k-1\}$ lexicographically. A string $w$ of length $n$ is defined to be in $L$ if and only if the rank of $w$ (in the lexicographic order) is greater than $r_n$. (Rank count starts at 1.) Clearly, the number of strings of length $n$ in $L$ is exactly $s_n$ = $k^n - r_n$. As shown in \cite{Beal}, $L$ is counting-regular and  $\overline{L}$ is not counting-regular. 

Finally, we provide a sketch of the proof that there is a deterministic polynomial time algorithm for $L$. Given a string $w$ of length $n$, and an integer $T$ (in binary) where the number of bits in $T$ is $O(n)$, it is easy to see that there is a deterministic algorithm that determines in time polynomial in $n$ if the rank of $w$ is greater than $T$. (This algorithm simply converts $T$ from binary to base $k$ and compares the resulting string to $w$ lexicographically. Base conversion can be shown to have complexity no more than that of integer multiplication.)
To complete the algorithm, we need to show how to compute $r_n$ (in binary), given $n$, in time polynomial in $n$. Note that $r_n$ is given by $b^{2n} \cos^2(n\theta)$. Clearly, $b^{2n}$ can be computed in polynomial time by repeated multiplication by $b$ . We don't even need repeated squaring to achieve a polynomial bound since the time complexity is measured in terms of $n$, not $log\ n$. Also,  $\cos \ (n \ arccos (a/b))$ can be computed as follows: $\cos \ (n \ arccos (a/b))$ is the well-known Tchebychev polynomial $T_n(a/b)$ which is explicitly given by the series \cite{Mason}:
$$\sum_{m=0}^{\lfloor n/2 \rfloor} {n \choose {2m}} \left ( a \over b \right )^{n-2m} \left( \left ( a \over b \right )^2 -1 \right)^m. $$

Since each of the terms in the above series can be computed in time polynomial in $n$, and since there are $\lfloor {n \over 2} \rfloor$ terms in the series, it is clear that $r_n$ can be computed in time polynomial in $n$. 
\qed \end{proof}

It is evident from Corollary \ref{cortomain} and closure of reversal-bounded $\DPDA$'s
under complement that counting-regular reversal-bounded $\DPDA$'s are closed under complement. For $\DPDA$'s generally,
it is open whether counting-regular $\DPDA$'s are closed under complement.
 At the end of the proof of Theorem \ref{main2}, we observed that the mapping that we used to map the strings from a language $L$ accepted by an unambiguous reversal-bounded $\NTM$ to a regular language increased the size of the alphabet from $k$ to $k+1$. If for every counting-regular $\DPDA$, the size of the alphabet of the simulating $\NFA$ is $k$, then it will follow that counting-regular $\DCFL$'s are closed under complement.

In fact, we can show the following stronger claim.

\begin{theorem}
The following statements are equivalent:
\begin{enumerate}
\item For every counting-regular $\DCFL$ over a $k$-letter alphabet, there is a regular language $L'$ over a $k$-letter alphabet such that $f_L(n)$ = $f_{L'}(n)$ for all $n$.

\item Counting-regular $\DCFL$'s are closed under complement.
\end{enumerate}
\end{theorem}
\begin{proof} (1) $\Rightarrow$ (2) is immediate from the above discussion.

To show that (2) $\Rightarrow$ (1), let $L$ be a counting-regular $\DCFL$ over a $k$-letter alphabet. This means $f_L(n)$ is a regular sequence. By (2), the complement of $L$ is counting-regular, so $k^n - f_L(n)$ is also a regular sequence. From Theorem \ref{Beal}, it follows that there is a regular language $L'$ over a $k$-letter alphabet such that $f_L(n)$ = $f_{L'}(n)$ for all $n$. 
\qed \end{proof}

Our conclusion is that the class of counting-regular languages (or counting-regular $\CFL$s) are very fragile in that it is not closed under basic operations such as union or intersection with regular languages. We conjecture that the counting-regular $\CFL$s are also not closed under Kleene star.

\section{Some Decision Problems Related to Counting-Regularity}
\label{sec:decidability}

In this section, some decision problems in regards to counting-regularity are addressed.
In particular, we will show the following:
(1) It is undecidable,
given  a real-time $1$-reversal-bounded $2$-ambiguous $\PDA$ $M$,
whether $L(M)$ (resp.\ $\overline{L(M)}$)
is counting-regular;
(2) It is undecidable, given a real-time $\NCM(1,1)$
 $M$, whether $L(M)$ (resp.\ $\overline{L(M)}$)
is counting-regular. 

We begin with the following result:









\begin{theorem}
\label{undecide1}
It is undecidable, given two real-time $1$-reversal-bounded $\DPDA$'s $M_1$ and $M_2$ accepting strongly counting-regular languages, whether $L(M_1) \cap L(M_2)$ is counting-regular. 
\end{theorem}
\begin{proof} 

Let $Z$ be a single-tape $\DTM$ working on an initially blank tape. We assume that if $Z$ halts on blank tape, it makes $2k$ steps for some $k \ge 2$.
We assume that $Z$ has one-way infinite tape and does not write blank. Let
$$L_1' = \{ \begin{array}[t]{l} ID_1 \# ID_3 \# \cdots  \# ID_{2k-1} \$  ID_{2k}^R \# \cdots \# ID_4^R \# ID_2^R \mid k \ge 2, \mbox{~each~} ID_i\\
\mbox{is a configuration of~} Z,  ID_1 \mbox{~is initial~} ID \mbox{~of~} Z \mbox{~on blank tape,}   \\
ID_{2k} \mbox{~is a unique halting~} ID,  ID_i \Rightarrow ID_{i+1} \mbox{~for~} i = 1, 3, \ldots, 2k-1 \},\end{array}$$
$$L_2' = \begin{array}[t]{l} \{ ID_1 \# ID_3 \# \cdots \# ID_{2k-1} \$  ID_{2k}^R \# \cdots  \# ID_4^R \# ID_2^R \mid k \geq 2, \mbox{~each~} ID_i\\
\mbox{is a configuration of~} Z,  ID_1 \mbox{~is initial~} ID \mbox{~of~} Z \mbox{~on blank tape}, \\
ID_{2k} \mbox{~is a unique halting~} ID,  ID_i  \Rightarrow ID _{i+1} \mbox{~for~} i = 2, 4, \ldots, 2k-2 \},\end{array}$$
with $L_1',L_2' \subseteq \Sigma^*$.
Clearly, $L_1'$  and $L_2'$ can be accepted by real-time $1$-reversal-bounded $\DPDA$'s $M_1'$ and $M_2'$. Moreover, $L_1' \cap L_2'$ is empty or
a singleton (if and only if $Z$ accepts, which is undecidable).
Let $a$, $b$, $1$ be three new symbols. Let
\begin{eqnarray*}
L_1 &=& \{ x w 1^n\ |\ x \in (a+b)^+, \ w \in L(M_1'),\ |x|_a = n \},\\
L_2 &=& \{x w 1^n\ |\  x \in (a+b)^+, \ w \in L(M_2'),\ |x|_b = n \}.
\end{eqnarray*}
We can construct a real-time $1$-reversal-bounded $\DPDA$ $M_1$ accepting $L_1$
as follows: $M_1$ reads $x$ and stores the number of $a$'s in the stack.
Then it simulates the $1$-reversal-bounded $\DPDA$ $M_1'$ on $w$, and
finally checks (by continuing to pop the stack) that the number
of $a$'s stored in the stack is $n$. Similarly, we can construct a 
real-time $1$-reversal-bounded $\DPDA$ $M_2$ accepting $L_2$.
By Corollary \ref{cortomain}, $L_1$ and $L_2$ are
strongly counting-regular.

Clearly, if $L(M_1') \cap L(M_2')$ = $\emptyset$,
then $L_1 \cap L_2$ = $\emptyset$,
hence $L_1 \cap L_2$ is counting-regular.

On the other hand, if $L(M_1') \cap L(M_2')$ is not empty, the intersection
is a singleton $w$ and thus $L_1 \cap L_2$ is given by:
$$L_1 \cap L_2 = \{ x w 1^n\ | \ x  \in (a+b)^+ ,\ |x|_a = |x|_b = n\}.$$

As in the proof of Theorem \ref{counter-examples}(ii), we can show that
$L_1 \cap L_2$ is not counting-regular.  In fact, if $|w|$ = $t$, then the 
number of strings of length $3n+t$ is ${2n}\choose{n}$ from which we can explicitly construct the generating function $f(z)$ for $L_1 \cap L_2$ as:
$$f(z) = {1 \over 3} \left ( {1 \over {\sqrt{1-4z^3}}} + {1 \over {\sqrt{1-4{\omega}z^3}}}  +  {1 \over {\sqrt{1-4{\omega^2}z^3}}} \right ) z^t.$$
Clearly $f(z)$ is not rational and the claim follows.
\qed \end{proof}

This result is used within the next proof:

\begin{theorem} \label{thm16}
It is undecidable, given a $\PDA$ $M$ that is real-time $2$-ambiguous and $1$-reversal-bounded,
whether $\overline{L(M)}$ is counting-regular.
Also, it is undecidable, given such a machine $M$,
whether $L(M)$ is counting-regular.
\end{theorem}
\begin{proof}
Consider the languages $L_1$ and $L_2$ in the proof of Theorem \ref{undecide1}.  Since $L_1$ and $L_2$
can be accepted by real-time $1$-reversal-bounded $\DPDA$'s,
$\overline{L_1}$ and $\overline{L_2}$ 
can also be accepted by real-time $1$-reversal-bounded $\DPDA$'s.
Hence $L = \overline{L_1} \cup \overline{L_2}$ 
can be accepted by a real-time $1$-reversal-bounded $2$-ambiguous $\PDA$ $M$.
It follows that $\overline{L(M)}$ is
counting-regular if and only if
$L_1 \cap L_2$ is counting-regular, which is undecidable
from Theorem \ref{undecide1}.

Now consider $L(M)$.  If $Z$ does not halt on blank tape,
then $L(M) = \Sigma^*$, which is counting-regular.
If $Z$  halts on blank tape,
 then we will show that $L(M)$ is not counting-regular as follows:
the generating function $g(z)$ of $L(M)$ is given by 
${1 \over {1-sz}} - f(z)$ where $f(z)$ is the generating function of $\overline{L(M)}$, and $s$ is the size of the alphabet over which $M$ is defined. If $L(M)$ is counting-regular, then $g(z)$ is rational, then so is $f(z)$ = ${1 \over {1-kz}} - g(z)$, contradicting the fact that $f(z)$ is not rational.
Hence, $L(M)$ is counting-regular if and only if $Z$ does not halt.
 This completes the proof. 
\qed \end{proof}

Next, we will consider
$\NCM(1,1)$ ($1$-reversal-bounded $1$-counter machines).
\begin{theorem}
It is undecidable, given a real-time $\NCM$(1,1) $M$, whether $\overline{L(M)}$ is counting-regular. Also, it is undecidable, given such a machine $M$, whether $L(M)$ is counting-regular.
\end{theorem}
\begin{proof} 
Again, we will use the undecidability of the halting problem for a $\DTM$  $Z$ 
on an initially blank tape. As before, assume that $Z$ has a one-way infinite tape
and does not write blank symbols. Represent each $ID$ (configuration) of $Z$ with 
blank symbols filled to  its right, since for the  languages we will define below, we
require that all $ID_i$'s have the same length.  So, e.g,  $ID_1 = q_0  B \cdots B$  
(where $B$ is the blank symbol).  We also require that the halting $ID$ have all non-blanks in state $f$, which is unique.  Clearly since 
$Z$ does not write any blank symbols,  if $Z$ halts on the blank tape, the lengths of the $ID$'s in
the halting sequence of $ID$'s do not decrease in length.  Assume that if $Z$ halts, 
it halts after $k$ steps for some $k \ge 2$.  

Let
$$L_1 = \begin{array}[t]{l} \{ ID_1 \# ID_3 \# \cdots \# ID_{2k-1} \$  ID_{2k}^R \# \cdots  \# ID_4^R \# ID_2^R \mid k \ge 2, \mbox{~each~} ID_i\\
\mbox{is a configuration of~} Z,  ID_1 \mbox{~is initial~} ID \mbox{~of~} Z \mbox{~on blank tape,~} 
ID_{2k} \mbox{~is the}\\ \mbox{halting~} ID,  ID_i  \Rightarrow ID_{i+1} \mbox{~for~} i = 1, 3, \ldots, 2k-1, |ID_1| = \cdots = |ID_{2k}|\},\end{array}$$
$$L_2 = \begin{array}[t]{l}\{ ID_1 \# ID_3 \# \cdots  \# ID_{2k-1} \$  ID_{2k}^R \# \cdots  \# ID_4^R \# ID_2^R \mid k \ge 2, \mbox{~each~} ID_i \mbox{~is} \\
\mbox{~a configuration of~} Z,  ID_1 \mbox{~is initial~} ID \mbox{~of~} Z \mbox{~on blank tape}, ID_{2k} \mbox{~is the}\\
\mbox{halting~} ID,  ID_i  \Rightarrow ID _{i+1} \mbox{~for~} i = 2, 4, \ldots, 2k-2, |ID_1| = \cdots = |ID_{2k}|\},\end{array}$$
with $L_1,L_2 \subseteq \Sigma^*$.
Note that $ID_{2k}$ must have all non-blanks in
state $f$.

Let $a, b, 1$ be new symbols.  We can construct a real-time
$\NCM$(1,1)  $M_1$ which operates as follows. When given input string $z$,
$M_1$ nondeterministically  selects one of the following tasks to execute:

\begin{enumerate}
\item $M_1$ checks and accepts if  $z$ is not a string of the form $x w 1^n \in (a+b)^+ \Sigma^+ 1^+$.
(This does not require the use of the counter.)

\item $M_1$ checks and accepts if $z$ is of the form $x w 1^n \in (a+b)^+ \Sigma^+ 1^+$
but $|x|_a \ne n$.  (This requires only one counter reversal.)

\item  $M_1$ checks that $z$ is of the form $x w 1^n \in (a+b)^+ \Sigma^+ 1^+$, but $w$ 
is not a string of the form in $L_1$.  $M_1$ does not check the lengths of the
$ID$'s and whether $ID_{i+1}$ is a successor of $ID_i$.  (This does not require
a counter.)

\item $M_1$ checks that $z$ is of the form $x w 1^n \in (a+b)^+ \Sigma^+ 1^+$
and $$w =  ID_1 \# ID_3 \# \ \cdots 
\# ID_{2k-1}  \$ ID_{2k}^R \# \ \cdots \ \# ID_4^R \# ID_2^R,$$ for
some $k \ge 2$  but $|ID_i| \ne |ID_j|$ for some $i \ne j$, or $ID_1$ is not the initial $ID$,
or $ID_{2k}$ is not the halting $ID$. (This requires one counter reversal.)

\item $M_1$ assumes that $z$ is of the form $x w 1^n \in (a+b)^+ \Sigma^+ 1^+$
and $$w = ID_1 \# ID_3 \#  \cdots 
\# ID_{2k-1} \$  ID_{2k}^R \#  \cdots  \# ID_4^R \# ID_2^R,$$ for
some $k \ge 2$ , $|ID_1| =   \cdots  = |ID_{2k}|$, $ID_1$ is the initial $ID$, $ID_{2k}$
is the halting $ID$, and accepts if $ID_{i+1}$  is not the successor of $ID_i$ for
some $i$ = $1, 3, \ldots , 2k-1$. Since all the $ID$'s are assumed to have the same length,
$M_1$ needs to only use a counter that reverses once to check one of the conditions.
\end{enumerate}
Similarly, we can construct a real-time $\NCM$(1,1) $M_2$ as above using $L_2$.
Let $M$ be a real-time $\NCM(1,1)$ accepting $L(M_1) \cup  L(M_2)$ and consider $\overline{L(M)}$.

If $Z$ does not halt on blank tape,
then $\overline{L(M)} = \overline{L(M_1) \cup L(M_2)} = \overline{L(M_1)}
\cap \overline{L(M_2)} = \emptyset$, which is counting-regular.

If $Z$ halts on blank tape, then
$\overline{L(M)}$ = $\{ x w 1^n\  |\  x \in (a+b)^+ ,\ |x|_a = |x|_b = n\}$ for some $w$
which is not counting-regular (as shown in Theorem \ref{undecide1}). 

Now consider $L(M)$.
If $Z$ does not halt on blank tape, then $L(M)$ = $\Sigma^*$, hence is counting-regular. If $Z$ halts on blank tape, then
$L(M)$ = $L(M_1) \cup  L(M_2)$, which we will show to be not counting-regular as follows: the generating function $g(z)$ of $L(M)$ is given by 
${1 \over {1-sz}} - f(z)$ where $f(z)$ is the generating function of $\overline{L(M)}$, and $s$ is the size of the alphabet over which $M$ is defined. If $L(M)$ is counting-regular, then $g(z)$ is rational, then so is $f(z)$ = ${1 \over {1-kz}} - g(z)$, contradicting the fact that $f(z)$ is not rational. 
\qed \end{proof}







\noindent
Note that the machine $M$ constructed in the proof above
is 1-reversal-bounded but not finitely-ambiguous.
Next, it is shown to be undecidable for $2$-ambiguous  machines but without the reversal-bound.

\begin{theorem} \label{thmXX}
It is undecidable, given a one-way 2-ambiguous nondeterministic one counter machine 
$M$,
whether $\overline{L(M)}$ is counting-regular.
Also, it is undecidable, given such a machine 
$M$,
whether $L(M)$ is counting-regular.
\end{theorem}
\begin{proof}
It is known that it is undecidable, given two deterministic
one counter machines (with no restriction on counter reversals) $M_1$ and
$M_2$, whether $L(M_1) \cap L(M_2) = \emptyset$ (shown implicitly in \cite{undecidablehartmanis}).  Moreover, if
the intersection is not empty, it is a singleton.
Let $\Sigma$ be the input alphabet of $M_1$ and $M_2$. Let $a, b, 1$ be three new symbols. Let
\begin{eqnarray*}
L_1 &=& \{ w x 1^n\ |\ w \in L(M_1), \ x \in (a+b)^+, \ |x|_a = n \},\\
L_2 &=& \{ w x 1^n\ |\ w \in L(M_2), \ x \in (a+b)^+, \ |x|_b = n \}.
\end{eqnarray*}
Clearly, we can construct deterministic one counter machines
accepting $L_1$ and $L_2$. Hence, 
$\overline{L_1}$ and $\overline{L_2}$ can be accepted by 
deterministic one counter machines.  It follows that  
$\overline{L_1} \cup \overline{L_2}$
can be accepted by a $2$-ambiguous nondeterministic one counter machine $M$.  Then, as in the
proof of Theorem \ref{thm16}, $\overline{L(M)}$ (resp., $L(M)$)
is counting-regular if and only if $L(M_1) \cap L(M_2)  = \emptyset$,
which is undecidable.
\qed
\end{proof}

In view of the above theorems, it is an interesting
open question whether the undecidability holds for reversal-bounded
finitely-ambiguous $\NCM$ machines.








\section{Slender Semilinear and Length-Semilinear Languages and Decidability Problems}
\label{sec:slender1}

A topic closely related to counting functions of formal languages
is that of slenderness. 
Decidability and closure properties of context-free languages ($\CFL$s) have
been investigated in \cite{Ilie,matrix,Salomaa,Ilie2,Honkala1998}.
For example, \cite{Ilie}
shows that it is decidable whether a $\CFL$ is slender, and
in \cite{matrix}, it is shown
that for a given $k \geq 1$, it is decidable whether a language generated by
a matrix grammar is $k$-slender (although here, the $k$ needs to be provided as input
in contrast to the $\CFL$ result).

In this section, we generalize these results to arbitrary language families that
satisfy certain closure properties.  These generalizations would then
imply the known results for context-free languages and matrix languages,
and other families where the problem was open.

First, we discuss the bounded language case. 
The constructive result of Theorem \ref{cor9} above, plus
decidability of slenderness for regular languages implies the following:
\begin{corollary} \label{cor13}
Let $\LL$ be a semilinear trio (with all properties effective). 
Then, it is decidable, given $L$ bounded in $\LL$ and words 
$u_1, \ldots, u_k$ such that $L \subseteq u_1^* \cdots u_k^*$,
whether $L$ is slender.
\end{corollary}
In \cite{Honkala1998}, the similar result was shown that 
it is decidable whether or not a given bounded semilinear language $L$
is slender.   

In our definition of a semilinear family of languages $\LL$, we only
require that every language in $\LL$ has a semilinear Parikh map.
However, it is known that in every semilinear trio $\LL$, all
bounded languages are bounded semilinear \cite{CIAA2016}, and therefore
the result of \cite{Honkala1998} also implies Corollary \ref{cor13}. Conversely, all
bounded semilinear languages are in the semilinear trio $\NCM$
\cite{ibarra1978,CIAA2016}; hence, given any bounded semilinear 
language (in any semilinear family
so long as we can construct the semilinear set), slenderness is decidable.
This method therefore also provides an alternate proof to the result
in \cite{Honkala1998}.

Next, we will examine the case where $L$ is not necessarily bounded.
One recent result is quite helpful in studying $k$-slender languages. In \cite{fullaflcounters}, the following was shown.
\begin{theorem}\cite{fullaflcounters}
Let $\LL$ be any semilinear full trio where the semilinearity and intersection with regular languages properties are effective. Then the smallest full
AFL containing intersections of languages in $\LL$ with $\NCM$,
denoted by $\hat{{\cal F}}(\LL \wedge \NCM)$, is effectively semilinear.
Hence, the emptiness problem for 
$\hat{{\cal F}}(\LL \wedge \NCM)$ is decidable.
\label{fullAFL}
\end{theorem}
We make frequent use of this throughout the proofs of the next two sections.

First, decidability of $k$-slenderness is addressed.
\begin{theorem}
\label{kslender}
Let $\LL$ be a full trio which is either:
\begin{itemize}
\item semilinear, or
\item is length-semilinear, closed under concatenation, and intersection with $\NCM$,
\end{itemize}
with all properties effective.
It is decidable, given $k$ and $L \in \LL$, whether $L$ is a k-slender language.
\end{theorem}
\begin{proof} Let $L \subseteq \Sigma^*$, and let $\#$ be a new symbol. 
First construct an $\NCM$ $M_1$
which when given $x_1 \# x_2\# \ldots \# x_{k+1}$, $x_i \in \Sigma^*, 1 \leq i \leq k+1$,
accepts if $|x_1|  = \cdots = |x_{k+1}|$, and $x_i\neq x_j$ are different, for
all $i \neq j$.  To do this, $M_1$ uses (many) counters to verify the lengths,
and it uses counters to guess positions and verify 
the discrepancies between $x_i$ and $x_j$, for each $i \ne j$.
Let $\LL' = \hat{{\cal F}}(\LL \wedge \NCM)$ (or just $\LL$ in the second case), which is semilinear
by Theorem \ref{fullAFL}, (or length-semilinear in the second case, by assumption).
Construct $L_2 \in \LL'$ 
which consists of all words of the form 
$x_1 \# x_2\# \ldots \# x_{k+1}$, where $x_i \in L$ (every full AFL is closed
under concatenation). 
Then $L_3 = L(M_1) \cap L_2 \in \LL'$.
Clearly, $L$ is not $k$-slender if and only if $L_3$ is not empty, 
which is decidable by Theorem \ref{fullAFL}.
\qed \end{proof}

There are many known semilinear full trios listed in Example \ref{semilinearfulltrioexamples}. Plus,
it is known that languages generated by matrix grammars form a length-semilinear (but not semilinear, in general) full
trio closed under concatenation and intersection with $\NCM$ (\cite{matrix},
where is it is shown that the languages are closed under intersection with
BLIND multicounter languages, known to be equivalent to $\NCM$ \cite{G78}). Therefore, the
result is implied for matrix grammars as well, although this is already known.
\begin{corollary}
\label{allcorollaries}
Let $\LL$ be any of the families listed in Example \ref{semilinearfulltrioexamples}. 
Then, the problem, ``for $k \geq 1$ and $L \in \LL$, is $L$ a $k$-slender language?''
is decidable.
\end{corollary}

All the machine models used in Example \ref{semilinearfulltrioexamples} have one-way inputs, however with a two-way input,
the problem is more complicated.
Now let $2\DCM(k)$ (resp., $2\NCM(k)$) be a two-way $\DFA$ (resp., two-way
$\NFA$) with end-markers on both sides of the input, augmented with $k$ reversal-bounded counters.
\begin{theorem} It is decidable, given $k$ and  $2\DCM(1)$ $M$, whether $M$ accepts 
a $k$-slender language.
\label{deckslender}
\end{theorem}
\begin{proof} We may assume that $M$ always halts \cite{IbarraJiang}. Given $M$,
construct another $2\DCM(1)$ $M'$ with a $(k+1)$-track tape. First  for each
$1 \le i  < k+1$, $M'$ checks that  the string in track $i$ is different from
the strings in tracks $i+1, \ldots, k+1$. Thus, $M$ needs to make multiple
sweeps of the $(k+1)$-track input.

Then $M'$ checks that the string in each track is accepted.  Clearly, $L(M)$ is 
not $k$-slender if and only if $L(M')$ is not empty. The result follows, since
emptiness is decidable for $2\DCM(1)$ \cite{IbarraJiang}.
\qed \end{proof}

The above result does not generalize for $2\DCM(2)$:

\begin{theorem} The following are true:
\begin{enumerate} 
\item  It is undecidable, given $k$ and a $2\DCM(2)$ $M$, whether $M$ accepts a
$k$-slender language, even when $M$ accepts a letter-bounded language
that is a subset of $a_1^* \cdots a_r^*$ for given $a_1, \ldots, a_r$.
\item It is undecidable, given $k$ and a $2\DCM(2)$ $M$, whether $M$ accepts a slender
language, even when $M$ accepts a letter-bounded language.
\end{enumerate}
\end{theorem}
\begin{proof} It is known \cite{ibarra1978} that it is undecidable, given a $2\DCM(2)$
$M$ accepting a language that is in $b_1^* \cdots b_r^*$ for given $b_1, \ldots, b_r$, 
whether $L(M) = \emptyset$.  Let $c$ and $d$ be new symbols.
Construct another $2\DCM(2)$ $M'$ which when given
a string $w = b_1^{i_1} \cdots b_k^{i_r} c^i d^j$, simulates $M$ on 
$b_1^{i_1} \cdots b_k^{i_r}$, and when $M$ accepts, $M'$ accepts $w$.
Then $L(M')$ is not $k$-slender for any given $k$ (resp., not slender) if and only if $L(M)$
is not empty, which is undecidable.
\qed \end{proof}

Whether or not Theorem \ref{kslender} holds for $2\NCM(1)$ $M$ is open.
However, we can prove a weaker version using the fact that it is
decidable, given a $2\NCM(1)$ $M$ accepting a bounded language
over $w_1^* \cdots w_r^*$ for given $w_1, \ldots, w_r$, whether
$L(M) = \emptyset$ \cite{DangIbarra}.

\begin{theorem}  It is decidable, given $k$ and $2\NCM(1)$ $M$ that accepts a 
language over $w_1^* \cdots w_r^*$ for given $w_1, \ldots, w_r$, whether $M$ accepts
a $k$-slender language.
\end{theorem}
\begin{proof}
We construct from $M$ another $2\NCM(1)$ $M'$ which, when given 
$x_1 \# x_2\# \cdots  \# x_{k+1}$,  where each
$x_i$ is in $w_1^* \cdots w_r^*$ first checks that all $x_i$'s 
are different.  For each $i$, and $j = i+1, \ldots, k+1$, $M'$ guesses the position
of discrepancy between $x_i$ and $x_j$ and records this position in the counter
so that it can check the discrepancy. (Note that only one reversal-bounded
counter is needed for this.) Then $M'$ checks that each $x_i$ is accepted.  
Clearly, $M$ is not $k$ slender if and only if $L(M')$ is not empty, which is 
decidable, since the language accepted by $M'$ is bounded.
\qed \end{proof}

We can also prove some closure properties.  Here is an example:

\begin{theorem}
The following are true:
\begin{enumerate} 
\item Slender (resp., thin) $\NCM$ languages are closed under intersection.
\item Slender (resp., thin) $2\DCM(1)$ languages ($2\NCM(1)$ languages) are 
    closed under intersection.
\end{enumerate}
\end{theorem}
\begin{proof}
Straightforward since the families of languages above are
closed under intersection.
\qed \end{proof}

Deciding if an $\NCM$ language (or anything more general than $\CFL$s) is $k$-slender for some $k$ (where $k$ is not part of the input) is open,
although as we showed in Theorem \ref{deckslender}, for a given $k$, we can 
decide if an $\NCM$ or $\NPCM$ accepts a $k$-slender language. We conjecture
that every $k$-slender $\NPCM$ language is bounded. If  
this can be proven, we will also need an algorithm to determine words $w_1, \ldots,
w_r$ such that the language is  a subset of $w_1^* \cdots w_r^*$, which 
we also do not yet know how to do.

Let $c \ge 1$.  A $2\DCM(k)$ ($2\NCM(k))$ $M$ is $c$-crossing if  the number of times
the input head  crosses the boundary of any two adjacent cells of the input is
at most $c$. Then $ M$  is finite-crossing if it is $c$-crossing for some $c$. It is known
that a finite-crossing $2\NCM(k)$ can can be converted to an $\NCM(k')$ for some
$k'$ \cite{Gurari1981220}. Hence every bounded language accepted by any
finite-crossing $2\NCM(k)$ is counting-regular.  The next result shows  that this is
not true if the two-way input is unrestricted:
\begin{theorem} The letter-bounded  language $L = \{a^i b^{ij} ~|~ i, j \ge 1\}$
is accepted by a $2\DCM(1)$ whose counter makes only $1$ reversal, but $L$ is not
counting-regular.
\end{theorem}
\begin{proof}
A $2\DCM(1)$ $M$ accepting $L$ operates as follows, given input $a^ib^k$ ($i, k \ge 1$):
$M$ reads and stores $k$ in the counter.  Then it makes multiple sweeps on $a^i$
while decrementing the counter to check that $k$ is divisible by $i$.

To see that $L$ is not counting-regular, we note 
that, for any $n \geq 2$, the number of strings of length $n$ in $L$ is $\phi(n)$, Euler's totient function 
is equal to the number of divisors of $n$. In fact, let the divisors of $n$ be
$d_1, d_2, \ldots, d_m$. Then, there are exactly $m$ strings of length $n$, namely: $a^{d_1} b^{(n/d_1 - 1)d_1}, a^{d_2} b^{(n/d_2 - 1)d_2}, \ldots , \\
a^{d_m} b^{(n/d_m - 1)d_m}$. Conversely, for each string $w$ of length $n$ in $L$, there is a unique divisor of $n$ (namely the number of $a$'s in $w$) associated with the string. This means the generating function of $L$ is $f(z)$ = $\sum_{n \geq 2} \phi(n) z^n$. 

It can be shown that $f(z)$ = $\sum_{n \geq 2} a_n {z^n \over {1-z^n}}$ where $a_n$ = $\sum_{d|n} \phi(d) \mu(n/d)$ where $\mu$ is the Mobius function.
From this expression, it is clear that every solution to $z^n=1$ is a pole of $f(z)$ and so the $n$'th root of unity is a pole of $f(z)$
for each positive integer $n \geq 2$. Since a rational function can have only a finite number of poles, it follows that $f(z)$ is not
rational and hence $L$ is not counting-regular.
\qed \end{proof}
However, it is known that  unary languages accepted by $2\NCM(k)$s  are
regular \cite{IbarraJiang}; hence, such languages are counting-regular.

This is interesting, as $2\DCM(1)$ has a decidable emptiness and slenderness problem, yet
there is a letter-bounded language from the theorem above accepted by a $2\DCM(1)$
that is not counting-regular.

The following result provides an interesting contrast, as it
involves a model with an undecidable membership
(and emptiness) problem, but provides an example of a (non-recursively enumerable) 
language that is all of slender, thin,
bounded, semilinear, but also counting-regular. 

\begin{theorem}
There exists a language $L$ that is letter-bounded and semilinear and thin and counting regular but is not recursively enumerable. Moreover, we can effectively construct a 
$\DFA$ $M$ such that
$f_{L(M)}(n) = f_L(n)$. 

\end{theorem}
\begin{proof}
Let $L \subseteq a^*$ be a unary language that is not recursively enumerable, which is known to exist \cite{Minsky}. Assume without loss of generality that the empty word is not in $L$, 
but the letter $a$ itself is in $L$.

Let $L'$ be the language consisting of, for each $n \geq 1$, the single word
with all $a$'s except for one $b$ in the position of the largest $m \leq n$ such that
$a^m \in L$.

Then $L'$ has one word of every length, and is therefore thin. Also, it has one word of every length and exactly one $b$ in every word, so it has the same Parikh map as $a^* b$, so is semilinear. Also, it is clearly bounded in $a^*b a^*$. It is also not recursively enumerable, otherwise if it were, then make a gsm \cite{Hopcroft} that outputs $a$'s for every $a$ until a $b$, then it outputs $a$. Then, for every remaining $a$, it outputs the empty word. Applying this to $L'$ gives $L$ (the position of the $b$ lets the gsm recover the words of $L$). But the recursively enumerable languages are closed under gsm mappings, a contradiction.
\qed \end{proof}

\section{Characterization of $k$-Slender Semilinear and Length-Semilinear Languages}
\label{sec:slender2}

This section discusses decidability properties (such as the problem of testing whether two
languages are equal, or one language is contained in another) for $k$-slender languages in
arbitrary families of languages satisfying certain closure properties. 
It is known that the equivalence problem for $\NPCM$ languages 
that are subsets of $w_1^* \cdots w_r^*$ for given $w_1, \ldots ,  w_r$ 
is decidable. However, as mentioned above, we do not know yet 
if $k$-slender languages are bounded and even if they are, we do not
know yet how to the determine the associated words $w_1, \ldots, w_r$.
Hence, the definition and results below are of interest.

The following notion is useful for studying decidability properties
of slender languages.
Let $k \ge 1$ be given.  A language $L$ is $k$-slender effective if we 
can effectively construct a $\DFA$ over a unary alphabet $\{1\}$ with $k+1$
distinguished states $s_0, s_1, \ldots, s_k$ (other states can exist) 
which, when given an input
$1^n$ where $n \ge 0$, halts in state $s_i$ if $f_L(n) = i$,  where  
$0 \le i \le k$. (Hence, the $\DFA$ can determine the number of strings 
in $L$ of length $n$, for every $n$.)

For example, consider the language
  $$L = \{a^i b^i \mid i \ge 1\}  \cup  \{c^i a^i d^i \mid i \ge 1\}.$$
Then,
$$f_L(n) = 
\begin{cases} 0 & \mbox{if $n = 0$ or $n=1$ or not divisible by $2$ and $3$,}\\
    1 & \mbox{if $n \ge 2$, is divisible by $2$, and not divisible by $3$,}\\
     1 & \mbox{if $n \ge 3$, is divisible by $3$, and not divisible by $2$,}\\
     2 & \mbox{if $n \ge 2$, is divisible by $2$ and divisible by $3$}.
\end{cases}$$
Clearly, $L$ is $2$-slender effective.

Next, we will discuss which $k$-slender languages are $k$-slender effective.
First, we say that the length-semilinear property is effective if it is possible
to effectively construct, for $L \in \LL$, a $\DFA$ accepting $\{1^n \mid f_L(n) \geq 1 \}$.
\begin{theorem}
\label{finiteunion} Let $\LL$ be an effective length-semilinear trio.
A finite union of thin languages in $\LL$ is $k$-slender effective.
\end{theorem}
\begin{proof}
Let $L$ is the finite union of $L_1, \ldots, L_k \in \LL$, each thin. Then construct
$h(L_i)$, where $h$ maps $L_i$ onto the single letter $1$. Then $h(L_i) = \{1^n \mid f_{L_i}(n) = 1\}$. Then since $\LL$ is length-semilinear, each of $h(L_i)$ are regular, and can be accepted by a $\DFA$ $M_i$. Thus, make another $\DFA$ $M'$ such that, on input $1^n$, $M'$ runs each $M_i$ in parallel on $1^n$, and then switches to distinguished state $s_j$ if there are $j$ of the $\DFA$s $M_1, \ldots, M_k$ that are accepting.
\qed \end{proof}

It will be seen next that for all $k$-slender languages in ``well-behaved''
families, they are $k$-slender effective. First, the following lemma is needed.
\begin{lemma} 
\label{makeeffective}
Let $\LL$ be a full trio which is either:
\begin{itemize}
\item semilinear, or
\item is length-semilinear, closed under concatenation, and intersection with $\NCM$,
\end{itemize}
with all properties effective.
Let $ k \ge 1$ and $L \in \LL$, a $k$-slender
language such that $f_{L}(n)$ is equal to $0$ or $k$ for every $n$.  
There
is a $\DFA$ that can determine $f_L(n)$. Hence $L$ is a $k$-slender
effective language.
\end{lemma}
\begin{proof}
Let $\LL' = \hat{{\cal F}}(\LL \wedge \NCM)$ (or just $\LL$ in the second case), which is semilinear
by Theorem \ref{fullAFL}, (or length-semilinear by assumption).
 
Consider
$L' = \{ x_1 \# \cdots \# x_k \mid x_1, \ldots, x_k \in L\} \in \LL'$.
Create $L''$ by intersecting $L'$ with an $\NCM$ language that enforces that all words
of the form 
$ x_1 \# \cdots \# x_k $ have
$|x_1| = \cdots = |x_k|$, and $x_i \ne x_j$ for each $i \ne j$.
Thus $L'' = \{  x_1 \# \cdots \# x_k  \mid  x_1, \ldots, x_k \in L,
|x_1| = \cdots = |x_k|, x_i \ne x_j \mbox{~for each~} i \ne j\}$.
Hence, $L''\in \LL'$. Let $L'''$ be the language obtained from $L''$ by
homomorphism that projects onto the single letter $1$. 
Since $L'''$ is length-semilinear, it can be accepted by a $\DFA$. Moreover, the length $n$ of a word
$x_1\# \cdots \# x_k \in L''$ can be transformed into $|x_1|$ via $\frac{n-(k-1)}{k}$. Given the $\DFA$,
then another 
$\DFA$ can be built that can determine $f_L(n)$.
\qed \end{proof}

Then, the following is true.
\begin{theorem}
Let $\LL$ be a full trio which is either:
\begin{itemize}
\item semilinear, or
\item is length-semilinear, closed under concatenation, union, and intersection with $\NCM$,
\end{itemize}
with all properties effective.
If $L \in \LL$ be a $k$-slender language $L$, then
$L$ is $k$-slender effective.
\label{thm22}
\end{theorem}
\begin{proof} The case $k = 1$ is true by Theorem \ref{finiteunion}. 
 Assume by induction that
the theorem is true for $k \geq 1$.

Now consider an $L \in \LL$ that is a $(k+1)$-slender language,
$k \ge 1$.  Hence $f_{L}(n) \leq (k+1)$ for each $n \geq 0$. 

Let $\LL' = \hat{{\cal F}}(\LL \wedge \NCM)$ (or just $\LL$ in the second case), which is semilinear
by Theorem \ref{fullAFL}, (or length-semilinear by assumption).
Let $A = \{x_1 \# \cdots \# x_{k+1} \mid x_1, \ldots,x_{k+1} \in L\} \in \LL'.$
Then intersect
$A$ with an $\NCM$ that enforces that all words
$x_1 \# \cdots \# x_{k+1}$ have
$|x_1| = \cdots = |x_{k+1} |$, and $x_i \ne x_j$ for each $i \ne j$.
Let $A'$ be the resulting language.
Then
$A' = \{x_1 \# \cdots \# x_{k+1} \mid x_1, \ldots,x_{k+1}\in L \mbox{~such that~} 
               |x_1| = \cdots = |x_{k+1} |, x_i \ne x_j \mbox{~for each~} i \ne j\}.$
By Lemma \ref{makeeffective}, a $\DFA$ accepting $\{1^n \mid f_L(n) = k+1\}$
can be effectively constructed. Thus, a $\DFA$ accepting
$\{1^n \mid f_L(n) \neq k+1\}$ can also be constructed. Furthermore,
$B = \{w \mid w\in L, f_L(|w|) \neq k+1\} \in \LL'$.
Then, $B$ is $k$-slender and, hence $k$-slender effective by induction hypothesis.
Hence, $L$ is $(k+1)$-slender effective.
\qed \end{proof}

The proof of Theorem \ref{thm22} actually shows the following:
\begin{corollary} 
Let $\LL$ be a full trio closed under concatenation, union, and intersection with $\NCM$, and is length-semilinear
with all properties effective.
Let $k \ge 1$. A language $L \in \LL$ is a $k$-slender language 
if and only if $L = L_1 \cup \cdots \cup L_k$, where  for $1 \le i \le k$, $L_i$ is
an $i$-slender effective language such that $f_{L_i}(n)$ is equal to $0$ or $i$ for each $n$.
\end{corollary}

Next, decidability of containment is addressed.
\begin{theorem} 
Let $\LL$ be a full trio which is either:
\begin{itemize}
\item semilinear, or
\item is length-semilinear, closed under concatenation, and intersection with $\NCM$,
\end{itemize}
with all properties effective.
It is decidable, given $L_1, L_2 \in \LL$ with $L_2$
being a $k$-slender language, whether 
$L_1 \subseteq L_2$.
\label{containment}
\end{theorem}
\begin{proof} Then $L_2$ is $k$-slender effective by Theorem \ref{thm22}. Without loss of generality, assume that the input alphabet 
of both $L_1$ and $L_2$ is $\Sigma$.  Let $1, \#$, and $\$$ be new symbols. 
Let $\LL' = \hat{{\cal F}}(\LL \wedge \NCM)$ (or just $\LL$ in the second case), which is semilinear
by Theorem \ref{fullAFL}, (or length-semilinear by assumption).
We will construct a sequence of machines and languages below.
\begin{enumerate}
\item First, let $M_1'$ (resp.\ $M_2'$) be the unary $\DFA$ accepting all words $1^n$ where
a word of length $n$ is in $L_1$ (resp.\ in $L_2$). Let $A_1 = L(M_1) - L(M_2)$. (This is
empty if and only if all lengths of words in $L_1$ are lengths
of words in $L_2$. This language is regular.

\item Construct $A_2 \in \LL'$ consisting of all words
$w = 1^n \$ x \$ y_1\# \cdots y_r\$$, where 
$x \in L_1$ and each $y_j \in L_2$.

\item Construct an $\NCM$ $A_3$ which, when given 
$w = 1^n\$x\$y_1\# \cdots y_r\$$, accepts $w$ if the following is true:
\begin{enumerate}
\item $r = f_{L_2}(n)$ (which can be tested since $L_2$ is $k$-slender effective). 
\item $|x| = |y_1| = \cdots = |y_r| = n$.
\item $y_i \ne y_j$ for each $i \ne j$.
\item $x \ne y_i$ for each $i$.
\end{enumerate}
Note that $A_3$ needs multiple reversal-bounded counters to carry out the four
tasks in parallel.
\item Construct $A_4 = A_2 \cap A_3 \in \LL'$.
\item Finally construct an $A_5 = A_4 \cup A_1 \in \LL'$ (full trios are closed
under union with regular languages \cite{G75}).
\end{enumerate}

It is easy to verify that $L_1 \not\subseteq L_2$ if and only 
if $A_5$ is not empty, which is decidable, since emptiness is decidable.
\qed \end{proof}
\begin{corollary}
Let $\LL$ be a full trio which is either:
\begin{itemize}
\item semilinear, or
\item is length-semilinear, closed under concatenation, and intersection with $\NCM$,
\end{itemize}
with all properties effective.
It is decidable, given $L_1,L_2 \in \LL$ that are 
$k$-slender languages, whether $L_1 = L_2$.
\end{corollary}

There are many semilinear full trios in the literature for which the properties in this section hold. 
\begin{corollary}
Let $\LL$ be any of the families from Example \ref{semilinearfulltrioexamples}.
The following are decidable:
\begin{itemize}
\item For $L_1, L_2$ with $L_2$ being $k$-slender, is $L_1 \subseteq L_2$?
\item For $L_1, L_2$ being $k$-slender languages, is $L_1 = L_2$?
\end{itemize}
\end{corollary}

Furthermore, matrix grammars are an example of a length-semilinear \cite{DassowHandbook} full trio closed under concatenation, union, and intersection with $\NCM$ (although they do accept non-semilinear languages). We therefore
get all these properties for matrix grammars as a consequence of these proofs. However, this result is already known \cite{matrix}.

Using the ideas in the constructive proof of the theorem above,
we can also show:
\begin{theorem} Let $\LL$ be a union and concatenation closed length-semilinear full trio with all properties effective
that is closed under intersection with $\NCM$.
Let $L_1, L_2 \in \LL$ with $L_2$ a $k$-slender language.
Then $L_1-L_2 \in \LL$.  Hence, the complement of any $k$-slender language
in $\LL$ is again in $\LL$.
\label{difference}
\end{theorem}
\begin{proof} Let $\Sigma$ be the (without loss of generality) joint alphabet of $L_1$ and
$L_2$, and let $\Sigma'$ be the set of the
primed versions of the symbols in $\Sigma$. Let $\#$, and $\$$ be new symbols. 
Consider input 
\begin{equation}
w = x\$y_1\# \cdots \#y_r,
\label{w}
\end{equation} where $x$ is in $(\Sigma')^*$ and $y_1, \ldots, y_r$ 
are in $\Sigma^*$, for some $0 \leq r \leq k$. 
By Theorem \ref{thm22}, $L_2$ is $k$-slender effective. 
Let $M'$ be this unary $\DFA$ accepting all words of lengths in $L_2$.
Build an $\NCM$ $M''$ that on input $w$, verifies:
\begin{enumerate}
\item $r = f(n)$.
\item $|x|=|y_1| = \cdots = |y_r|$.
\item $y_i \ne y_j$ for each $i \ne j$.
\item $h(x) \ne y_i$ for each $i$, where $h(a') = a$ for each $a' \in \Sigma'$.
\end{enumerate}

Consider $L''' \in \LL$ 
consisting of all words of the form of $w$ in Equation \ref{w},
where $x \in L_1$, and each $y_i \in L_2$.
This is in $\LL$ since $\LL$ is closed under concatenation.

Now define a homomorphism $h_1$ which maps $\#, \$$, and symbols in $\Sigma$ to
$\epsilon$ and fixes letters in $\Sigma'$.  Clearly, $h_1(L''' \cap L(M''))$ is 
$L_1 - L_2$, and it is in $\LL$.
\qed \end{proof}

This holds for not only the matrix languages, but also concatenation and union-closed
semilinear full trios closed under intersection with $\NCM$. Some examples are:
\begin{corollary}
Let $\LL$ be any family of languages that are
accepted by a machine model in Example \ref{semilinearfulltrioexamples} 
that are augmented by reversal-bounded counters.
Given $L_1, L_2 \in \LL$ with $L_2$ being $k$-slender, then $L_1 - L_2 \in \LL$.
Furthermore, the complement of any $k$-slender language in $\LL$ is again in $\LL$.
\end{corollary}

Next, decidability of disjointness for $k$-slender languages will be addressed.
\begin{theorem}
Let $\LL$ be a full trio which is either:
\begin{itemize}
\item semilinear, or
\item is length-semilinear, closed under concatenation, union, and intersection with $\NCM$,
\end{itemize}
with all properties effective.
Given $L_1,L_2 \in \LL$ being $k$-slender languages,
it is decidable whether $L_1 \cap L_2 = \emptyset$.
\end{theorem}
\begin{proof}
Let $\LL' = \hat{{\cal F}}(\LL \wedge \NCM)$ (or just $\LL$ in the second case), which is semilinear
by Theorem \ref{fullAFL}, (or length-semilinear by assumption).

Notice that $L_1 \cap L_2 = (L_1 \cup L_2) - ((L_1 - L_2) \cup (L_2-L_1))$.
By Theorem \ref{difference}, $L_1 - L_2 \in \LL$ and $L_2 - L_1 \in \LL$, and
both must be $k$-slender since $L_1$ and $L_2$ are both $k$-slender.
Certainly $(L_1 - L_2) \cup (L_2-L_1) \in \LL$, and is also $2k$-slender.
Also, $L_1 \cup L_2 \in \LL$.
Hence, by another application of Theorem \ref{difference},
$(L_1 \cup L_2) - ((L_1 - L_2) \cup (L_2-L_1))\in \LL$. Since emptiness
is decidable in $\LL$, the theorem follows.
\qed \end{proof}
This again holds for all the families in Example \ref{semilinearfulltrioexamples} plus the languages accepted by matrix grammars.

An interesting open question is whether every $k$-slender $\NCM$ language (or other more
general families)
can be decomposed into a finite disjoint union of
thin $\NCM$ languages.
Although we have not been able to show this, we can give a related result. 
To recall, in \cite{Harju}, the model $\TCA$ is introduced consisting of a nondeterministic Turing machine with a one-way read-only input tape, a finite-crossing read/write tape, and reversal-bounded counters. It is shown that this model only accepts semilinear languages, and indeed, it is a full trio. Clearly, the model is closed under intersection with $\NCM$ by adding more counters. Although we do not know whether it is possible to decompose $\NCM$ slender languages into thin $\NCM$ languages, we can decompose them into thin $\TCA$ languages.

\begin{theorem}
Every $k$-slender $\NCM$ language $L$ is a finite union of thin $\TCA$ languages.
\end{theorem}
\begin{proof}
Let $M$ be an $\NCM$ accepting $L$. Since $L$ is $k$-slender, for each $n$, there are either exactly $k$ words of length $n$, or $k-1$ words of length $n$, etc.\ or $0$ words of length $n$. 
Let $A_k = \{ x_1 \# \cdots \# x_k \mid x_1, \ldots, x_k \in L(M), |x_1| = \cdots = |x_k|$, $x_1 < \cdots <x_k\}$ 
(the $<$ relation uses lexicographic ordering). For all such words $x_1 \# \cdots \# x_k$,
then the lengths of the first parts, $|x_1|$, are exactly those lengths $n$ such that $f_L(n) = k$. Then build a $\TCA$ $M'$ accepting $A_k$ as follows: $M'$ reads
$x_1 \# \cdots \# x_k$, and verifies each $x_i \in L(M)$ using a set of counters, while in parallel verifying their lengths are the same.
In parallel, $M'$ writes $x_1$ on the worktape; then when reading $x_2$ letter-by-letter, it scans $x_1$ on the worktape also letter-by-letter,
and in the first position where they differ, it verifies that $x_2>x_1$. From that point on, it replaces $x_1$ on the tape with $x_2$. It then
repeats up to $x_k$.

Let $G_i$ be a gsm that extracts the $i$'th ``component'' of $A_k$. Then $G_1(A_k), \ldots, G_k(A_k)$ are all thin languages.
As they are thin, there is a $\DFA$ $M_k$ accepting all these lengths of words. 
Next, let $A_{k-1} = \{ x_1 \# \cdots \# x_{k-1} \mid x_1, \ldots, x_{k-1} \in L(M), |x_1| = \cdots = |x_{k-1}|$, $x_1 < \cdots <x_{k-1}, 1^{|x_1|} \notin L(M_k)\}$. These must be those words $x_1 \# \cdots \# x_{k-1}$ where there cannot be a $k$'th such word (since $1^{|x_1|} \notin L(M_k)$).
And again, we can use $C_1(A_{k-1}), \ldots, C_{k-1}(A_{k-1})$ to separate based on lexicographic order, and each such language is thin. 

Continuing in this fashion down to $1$, we see that $L$ is a finite union of thin languages in $\TCA$ (for all $C_i(A_{j}), 1 \leq i \leq k$).
\qed \end{proof}

A language $L$ is $k$-counting-regular if there exists a regular language
$L'$ such that $f_{L'}(n) = f_L(n) \le k$ for $n \ge 0$ (this is equivalent to $L$ being $k$-slender and counting-regular).  $L$ is finite-counting
regular if it is $k$-counting regular for some $k$.
The next result shows that  all $k$-slender languages
in ``well-behaved'' language families are $k$-counting-regular.
\begin{theorem}
Let $\LL$ be a full trio which is either:
\begin{itemize}
\item semilinear, or
\item is length-semilinear, closed under concatenation, and intersection with $\NCM$,
\end{itemize}
with all properties effective.
Let $L$ be a $k$-slender language in $\LL$. Then $L$ is $k$-counting-regular (and thus counting-regular)
and we can effectively
construct a $\DFA$ $M$ such that $f_{L}(n) = f_{L(M)}(n) \le k$  for $n \ge 0$.
Moreover, $L(M)$ is bounded with $L(M) \subseteq 1^*\{\#_1, \ldots, \#_k\}$ for some distinct symbols $\#_1, \ldots, \#_k$.
\end{theorem}
\begin{proof}
From Lemma \ref{makeeffective}, $L$ is $k$-slender effective, and we can construct
a $\DFA$ $ M'$  such that  when given input $1^n$ , halts in state $s_i$ $(0 \leq i \le k)$
if $f_L(n) = i$.

Let $\#_1, \ldots, \#_k$ be new symbols.  Construct a $\DFA$ $M$ which, on input
$w = 1^{n-1}\#_s$ (for some $1 \le s \le k$) simulates  $M'$ on $w$. ($M$ pretends
that $\#_s$ is $1$ in  the simulation).  If $M'$  lands in state $s_t$   ($1 \le t \le k$),
$M$ accepts $w$  if and only of $ s = 1, \ldots ,t$. (Of course, if none of $M'$ lands in
state $s_0$, $M$ does not accept $w$).

It is easy to verify that $f_L (n) = L_{L(M)}(n) \le k$ for all $n \ge 0$.
Hence, $L$ is $k$-counting-regular.
\qed \end{proof}

\begin{corollary}
Let $\LL$ be a full trio which is either:
\begin{itemize}
\item semilinear, or
\item is length-semilinear, closed under concatenation, and intersection with $\NCM$,
\end{itemize}
with all properties effective.  Then
\begin{enumerate}
\item
$L$ is $k$-slender if and only if $L$ is $k$-counting regular.
\item
$L$ is slender if and only if $L$ is finite-counting regular.
\end{enumerate}
\end{corollary}

From Corollary \ref{cor13} and the above corollary:

\begin{corollary} 
Let $\LL$ be a semilinear trio (with all properties effective). 
Then, it is decidable, given $L$ bounded in $\LL$ and words 
$u_1, \ldots, u_k$ such that $L \subseteq u_1^* \cdots u_k^*$,
whether $L$ is finite-counting-regular.
\end{corollary}

As in Section 4, we assumed in this section that the closure 
properties are effective, since we wanted the results to be effective.
However, we can remove this assumption and many of the results would
still hold existentially.  For example, the closure properties
would still hold.

\section{Conclusions}

In this work, we attempted to understand languages with simple counting functions: those that have counting functions that belong to the class of counting functions of regular languages (known
as counting-regular languages), and those for which the counting function is bounded by a constant $k$,
(known as $k$-slender languages). First, it is shown that all
unambiguous $\NTM$s with a one-way input and a reversal-bounded worktape are counting-regular. Then, certain ``well-behaved'' language
families are considered, that form semilinear full trios $\LL$. 
It is shown that the counting functions for the bounded
languages in $\LL$ coincide with the counting functions for the bounded regular languages.
Also, all $k$-slender languages in $\LL$ have the same counting function as
some bounded regular language. The
containment, equality, and disjointment problems are shown to be decidable
for $k$-slender languages in $\LL$. Most results are general enough to cover
more general families that are not semilinear should they satisfy certain other closure properties (such families include the languages generated by matrix grammars).

We conclude with some open problems arising from this study. 
It is open whether for every counting-regular $\CFL$ $L$ over a $k$-letter alphabet, there is a regular language $L'$ over a $k$-letter alphabet such that $f_L$ = $f_{L'}$. 
Regarding closure properties, we conjecture that the counting-regular $\CFL$s are not closed under Kleene star.
It is unknown whether every $k$-slender language is bounded within every semilinear full trio (this is true within
the context-free languages).  Also, the decidability status of whether a language $L$ in
an arbitrary full trio is slender ($k$ is not given) is open (this is again decidable
for the context-free languages). Similarly, it is unknown whether there
is a procedure to determine words $w_1, \ldots, w_n$, if they exist, such that $L \subseteq w_1^* \cdots w_n^*$, for $L$ in such full trios (this is again possible for context-free languages).

\section*{Acknowledgements}
We thank Flavio D'Alessandro for pointing out to us that Theorem \ref{cor9} is also implied by the main result in \cite{FlavioBoundedSemilinear}.

\bibliography{bounded}{}
\bibliographystyle{elsarticle-num}

\end{document}